\documentclass[twoside,leqno,twocolumn]{article}
\usepackage{ltexpprt}
\usepackage{epsfig,amssymb,amsmath,amsfonts,graphicx}
\usepackage{algorithm,algorithmic}
\usepackage{multirow}

\newtheorem{@problem}{Problem}[section]
 \newenvironment{problem}{\begin{@problem}}{\end{@problem}}

\begin{document}

\title{\Large A $\frac{5}{2}$-Approximation Algorithm for Coloring Rooted Subtrees of a Degree $3$ Tree}
\author{Anuj Rawat\thanks{Intel, Hillsboro, OR USA.}}
\date{}

\maketitle

\begin{abstract}
A rooted tree $\vec{R}$ is a rooted subtree of a tree $T$ if the tree obtained
by replacing the directed edges of $\vec{R}$ by undirected edges is a subtree
of $T$. We study the problem of assigning minimum number of colors to a given
set of rooted subtrees $\mathcal{R}$ of a given tree $T$ such that if any two
rooted subtrees share a directed edge, then they are assigned different colors.
The problem is NP hard even in the case when the degree of $T$ is restricted to
$3$ \cite{erlebach_hicss97}. We present a $\frac{5}{2}$-approximation algorithm
for this problem. The motivation for studying this problem stems from the
problem of assigning wavelengths to multicast traffic requests in all-optical
WDM tree networks.
\end{abstract}

\section{Introduction}
\label{sec:introduction}
\subsection{Motivation}
\label{subsec:motivation}
In Wavelength Division Multiplexing (WDM) \cite{garcia_book00} (p.208-211)
multiple signals are transmitted simultaneously over a single optical fiber by
using a different wavelength of light for each signal. The extremely high data
transfer rate achievable by WDM, along with the low bit error rate have made it
very attractive for backbone networks. Electronic switching becomes
prohibitively expensive at this high data rate and hence switching is typically
carried out in optical domain with the move between optical and electronic
domains restricted to the source and the destination nodes. This
scenario, in which a single \emph{lightpath} is constructed between the source
and the destination, is called \emph{transparent} or \emph{all-optical}
networking. In absence of wavelength converters, which is usually the case due
to their high cost, a lightpath must use the same wavelength on every fiber link on
which it exists. This is called the \emph{wavelength continuity constraint}.
Also, if two lightpaths share a fiber link (in the same direction), then they
must be assigned different wavelengths. In case of multicast traffic requests
(single source-multiple destinations), in order to maintain transparent optics,
network nodes capable of performing light splitting \cite{mukherjee_comm99}
and tap-and-continue operations \cite{deogun_jlt00} are employed. A single
\emph{light tree} is constructed from the source to the corresponding set of
destinations to support a multicast request. The light is split and sent onto
multiple fiber links on the nodes where bifurcation is required. On the intermediate
nodes that are also in the destination set, a small amount of light is tapped
and used to retrieve the data, while the rest of the light is allowed to travel
through. The wavelength continuity constraint requires that the light tree use
the same wavelength on every fiber link on which it exists. In case when
the underlying fiber network is a tree, the routing of the light trees
corresponding to the traffic requests is fixed and the given traffic requests
can be treated as rooted subtrees of the underlying fiber tree. So the problem
reduces to assigning a minimum number of wavelengths to these rooted subtrees
such that any two rooted subtrees sharing a directed edge are assigned
different wavelengths.

\subsection{Notations and Definitions}
\label{subsec:notation}
$|S|$ denotes the cardinality of a finite set $S$. For real valued $x$,
$[x]^+$ denotes $\max\{x,0\}$. $f(S)$ denotes the image of mapping
$f:D\longrightarrow R$ restricted to set $S\subseteq D$.

Unless otherwise stated, all graphs are assumed to be simple. For graph $G$,
$E_G$ and $V_G$ denote the edge set and
the vertex set, respectively. An edge between vertices $u,v\in V_G$ is denoted
by the set $\{u,v\}$. Similarly, for a directed graph $\vec{G}$,
$E_{\vec{G}}$ and $V_{\vec{G}}$ denote the set of directed edges and vertices, respectively.
For a pair of vertices $u,v\in V_{\vec{G}}$, a directed edge from $u$ to $v$ is
denoted by the ordered pair $(u,v)$.
$\bar{G}$ denotes the
complement of graph $G$. $G[W]$ denotes the subgraph of graph $G$ induced by
vertex set $W\subseteq V_G$.

The undirected multigraph obtained by replacing all the directed edges of
directed graph $\vec{G}$ by undirected edges is referred to as the
\emph{skeleton} of $\vec{G}$. A directed graph $\vec{R}$ is a \emph{rooted
tree} if (i) its skeleton is a tree, (ii) there is a unique vertex
$r\in V_{\vec{R}}$ with in-degree $0$, and (iii) every other vertex has
in-degree $1$.
A directed graph $\vec{R}$ is a \emph{rooted subtree} of tree $T$ if
$\vec{R}$ is a rooted tree and its skeleton is a subtree of $T$.
In this case, we also refer to $T$ as the \emph{host tree} of $\vec{R}$. Let
$\mathcal{R}$ be a multiset\footnote{For ease of exposition, in this paper we
use the term set even though the object being referred to might be a
multiset.} of rooted subtrees of tree $T$. We denote the set of all the rooted
subtrees in $\mathcal{R}$ that contain directed edge $(u,v)$ by
$\mathcal{R}[(u,v)]$. If a rooted subtree $\vec{R}$ contains directed edge
$(u,v)$, we say that it is \emph{present} on the directed edge $(u,v)$.
Moreover, the set $\mathcal{R}$ of rooted subtrees of the directed graph
$\vec{G}$ \emph{collide} on directed edge $(u,v)$, if for every rooted subtree
$\vec{R}\in\mathcal{R}$, $(u,v)\in E_{\vec{R}}$. If the directed edge on which
the collision occurs is not important for the subsequent discussion, we simply
say that the set of rooted subtrees collide. With a slight abuse of notation,
we denote the set of all rooted subtrees in $\mathcal{R}$ that contain either
directed edge $(u,v)$ or $(v,u)$ by $\mathcal{R}[\{u,v\}]$. The \emph{load} of
a set $\mathcal{R}$ of rooted subtrees on a tree $T$ is defined to be the
maximum number of rooted subtrees in that set that share a directed edge and is
denoted by $l_{\mathcal{R}}$.

Let $\mathbb{N}$ denote the set of natural numbers. A valid \emph{coloring} of
a given set of rooted subtrees $\mathcal{R}$ of a tree $T$ is a map
$\psi:\mathcal{R}\longrightarrow\mathbb{N}$ such that for any pair of rooted
subtrees $\vec{R}_i,\vec{R}_j\in \mathcal{R}$ that collide,
$\psi(\vec{R}_i)\neq\psi(\vec{R}_j)$. We denote the set of all valid colorings
by $\Psi_{\mathcal{R}}$. We can create a \emph{conflict graph} for a
given set of rooted subtrees $\mathcal{R}$ of a tree $T$ where the vertices
represent the rooted subtrees and there is an edge between two vertices if the
corresponding rooted subtrees collide. We denote this conflict graph by
$G_{\mathcal{R}}$. Note that coloring rooted subtrees $\mathcal{R}$ is
equivalent to coloring the vertices of $G_{\mathcal{R}}$.

\subsection{Problem Statement}
\label{subsec:problem}
The coloring problem that we are interested in is stated as Problem
\ref{prob:problem} below.
\begin{problem}
\label{prob:problem}
Given a set of rooted subtrees $\mathcal{R}$ of a tree $T$ with degree at most
$3$, find a coloring $\psi\in\Psi_{\mathcal{R}}$ that minimizes the number of colors used.
\end{problem}
Note that assigning wavelengths to a set of multicast traffic requests in an
all-optical WDM network where the underlying fiber topology is a tree $T$ is
exactly equivalent to determining the coloring for the corresponding set of
rooted subtrees $\mathcal{R}$.

\subsection{Related Work}
\label{subsec:related_work}
The work that is most closely related to the problem of coloring a given set of
rooted subtrees of a tree, consists of the following:
\begin{itemize}
\item Coloring a given set of undirected paths on a tree.
\item Coloring a given set of directed paths on a tree.
\item Coloring and characterization of a given set of subtrees of a tree.
\end{itemize}
Our contribution can be seen as the next logical step in this series of works.

In \cite{golumbic_jct85}, Golumbic et al. proved that determining a minimum
coloring for a given set of undirected paths on a tree is NP hard in general.
They showed that undirected path coloring in stars is equivalent to edge
coloring in multigraphs. Since edge coloring is NP hard \cite{holyer_sicomp81},
undirected path coloring in stars is also NP hard. In fact, as observed in
\cite{erlebach_tcs01}, this equivalence result has several important
implications:
\begin{itemize}
\item Undirected path coloring is solvable in polynomial time in bounded degree
trees.
\item Undirected path coloring is NP hard for trees of arbitrary degrees (even
with diameter 2).
\item Any approximation algorithm for edge coloring in multigraphs can be
modified into an approximation algorithm for undirected path coloring in
trees and vice versa with the same approximation ratio.
\item Approximating undirected path coloring in trees of arbitrary degree with
an approximation ratio $\frac{4}{3}-\epsilon$ for any $\epsilon>0$ is NP hard.
\end{itemize}
In \cite{tarjan_dm85}, Tarjan introduced a $\frac{3}{2}$-approximation
algorithm for coloring a given set of undirected paths in a tree. Later, this
ratio was rediscovered by Raghavan and Upfal \cite{raghavan_stoc94} in the
context of optical networks. Mihail et al. \cite{mihail_focs95} presented a
coloring scheme with an asymptotic approximation ratio of $\frac{9}{8}$.
Nishizeki et al. \cite{nishizeki_sidma00} presented an algorithm for edge
coloring multigraphs with an asymptotic approximation ratio of $1.1$ and an
absolute approximation ratio of $\frac{4}{3}$. This improves the asymptotic and
the absolute approximation ratio of undirected path coloring in trees to $1.1$
and $\frac{4}{3}$ respectively.

In \cite{erlebach_pasa96}, Erlebach et al. proved that coloring a given set of
directed paths in trees is NP hard. The hardness result holds even when we
restrict instances to arbitrary trees and sets of directed paths of load $3$ or
to trees with arbitrary degree and depth $3$ \cite{kumar_ipl97}. For this
problem, Mihail et al. \cite{mihail_focs95} gave a $\frac{15}{8}$-approximation
algorithm. This ratio was improved to $\frac{7}{4}$ in \cite{kaklamanis_esa96}
and \cite{kumar_soda97}, and finally to $\frac{5}{3}$ in
\cite{kaklamanis_icalp97}. All these are greedy, deterministic algorithms and
use the load of the given set of directed paths as the lower bound on the
number of colors required. In \cite{kaklamanis_icalp97}, Kaklamanis et al. also
proved that no greedy, deterministic algorithm can achieve a better
approximation ratio than $\frac{5}{3}$.

Unlike its undirected counterpart, Erlebach et al. \cite{erlebach_hicss97}
proved by a reduction from circular arc coloring that the problem of coloring
directed paths is NP hard even in binary trees. This result also implies that
the problem that we are interested in is also NP hard. In \cite{kumar_soda97},
Kumar et al. gave a problem instance where the given set of directed paths on a
binary tree of depth $3$ having load $l$ requires at least $\frac{5}{4}l$
colors. Caragiannis et al. \cite{caragiannis_wocs97} and Jansen
\cite{jansen_wocs97} gave algorithms for the directed path coloring
in binary trees having approximation ratio $\frac{5}{3}$ (same as for general trees). In
\cite{auletta_approx00}, Auletta et al. presented a randomized greedy algorithm
for coloring directed paths of maximum load $l$ in binary trees
of depth $O(l^{\frac{1}{3}-\epsilon})$ that uses at most $\frac{7}{5}l+o(l)$
colors. They also proved that with high probability, randomized greedy
algorithms cannot achieve an approximation ratio better than $\frac{3}{2}$ when
applied for binary trees of depth $\Omega(l)$, and $1.293-o(1)$ when applied for
binary trees of constant depth. Moreover, they proved an upper
bound of $\frac{7}{5}l+o(l)$ for all binary trees. In
\cite{erlebach_tcs01}, Erlebach et al. proved that approximating directed path
coloring in binary trees with an approximation ratio $\frac{4}{3}-\epsilon$ for
any $\epsilon>0$ is NP hard.

In \cite{jamison_dm05} Jamison et al. proved that the conflict graphs of
subtrees in a binary tree are chordal \cite{fulkerson_pjm65}, and therefore
easily colorable \cite{gavril_sicomp72}. In \cite{golumbic_eurocomb05} Golumbic
et al. proved that the conflict graphs (obtained as described above) of
undirected paths on degree $4$ trees are weakly chordal \cite{hayward_jct85},
therefore coloring them is easy \cite{hayward_talg07}. Later, in
\cite{golumbic_wg06}, they extended the result to the conflict graph of
subtrees on degree $4$ trees.

For an extensive compilation of complexity results on both directed and
undirected paths in trees from the perspective of optical networks, the reader
is referred to \cite{kumar_ipl97} and \cite{erlebach_tcs01}. And for a survey
of algorithmic results, the reader is referred to \cite{caragiannis_ci01},
\cite{caragiannis_stacs04} and \cite{caragiannis_chap06}.

Ours is the first work to study the problem of coloring rooted subtrees of a
tree. As stated previously, this can be seen as the
directed counterpart of the problem of coloring subtrees of a tree.

\section{Algorithm}
\label{sec:greedy}
In this section, we present a greedy coloring scheme for Problem
\ref{prob:problem}. The algorithm is presented as Algorithm
\ref{algo:greedycolor} (GREEDY-COL). We denote the coloring generated by this
algorithm as $\psi^{\mathrm{GDY}}$. The algorithm proceeds in \emph{rounds}. In
each round we select and \emph{process} a host tree edge which has not been
selected in any of the previous rounds. Processing a host tree edge means
coloring all the uncolored rooted subtrees present on that edge.

\begin{footnotesize}
\begin{algorithm}
\footnotesize
\caption{GREEDY-COL}
\label{algo:greedycolor}
\begin{algorithmic}[1]
\REQUIRE Set of rooted subtrees $\mathcal{R}$ on host tree $T$ with degree at
most $3$.

\ENSURE A coloring $\psi^{\mathrm{GDY}}\in\Psi_{\mathcal{R}}$.

\STATE Perform BFS on host tree $T$ starting with arbitrary vertex as the
root and enumerate tree edges in the order of their discovery. Let
$\{e_1,\dots,e_{|E_T|}\}$ be the ordered set of edges $E_T$.

\STATE $\mathcal{P}_0\leftarrow\emptyset$

\FOR{$i=1$ to $|E_T|$}

\STATE
$\mathcal{Q}_i\leftarrow\mathcal{R}[e_i]\setminus\mathcal{P}_{i-1}$\\

\IF{edge $e_i=\{u,v\}$ is of type $4$ as defined in Lemma \ref{lem:analysis1}}

\STATE Let
$\psi_1,\psi_2\in\Psi_{\mathcal{Q}_i\cup\mathcal{P}_{i-1}}$

\STATE
$\psi_1(\vec{R}),\psi_2(\vec{R})\leftarrow\psi^{\mathrm{GDY}}(\vec{R})$
for every $\vec{R}_j\in\mathcal{P}_{i-1}$ (uncolored otherwise).

\STATE
PROCESS-EDGE-1$(T,\{u,v\},\mathcal{P}_{i-1},\mathcal{Q}_i,\psi_1)$

\STATE
PROCESS-EDGE-2$(T,\{\{u,v\},\{u,w\},\{u,x\}\},\mathcal{P}_{i-1},\mathcal{Q}_i,\psi_2)$

\IF{$|\psi_1(\mathcal{P}_{i-1}\cup\mathcal{Q}_i)|\leq|\psi_2(\mathcal{P}_{i-1}\cup\mathcal{Q}_i)|$}\label{line:greedycolor_10}

\STATE
$\psi^{\mathrm{GDY}}(\vec{R})\leftarrow\psi_1(\vec{R})$
for every $\vec{R}\in\mathcal{Q}_i$

\ELSE

\STATE
$\psi^{\mathrm{GDY}}(\vec{R})\leftarrow\psi_2(\vec{R})$
for every $\vec{R}\in\mathcal{Q}_i$

\ENDIF

\ELSE

\WHILE{$\exists$ some uncolored $\vec{R}\in\mathcal{Q}_i$}

\STATE
$\psi^{\mathrm{GDY}}\!(\vec{R})\leftarrow\min\{l\in\mathbb{N}:
\nexists\,\vec{S}\!\in\!\mathcal{P}_{i-1}\!\cup\!\mathcal{Q}_i
\textrm{ s.t. } \vec{R},\vec{S} \textrm{ collide and }
\psi^{\mathrm{GDY}}\!(\vec{S})\!=\!l\}$\label{line:greedycolor_17}

\ENDWHILE

\ENDIF

\STATE $\mathcal{P}_i\leftarrow\mathcal{P}_{i-1}\cup\mathcal{Q}_i$

\ENDFOR
\end{algorithmic}
\end{algorithm}
\end{footnotesize}

\subsection{Edge Order}
\label{subsubsec:edgeorder}
We traverse the edges of the host tree in a breadth-first manner, i.e.,
starting with an arbitrary vertex as root, we perform a Breadth First Search
(BFS) on the host tree $T$ and rank its edges in the order of their discovery,
and then process the edges in this order.\footnote{Note that this edge ordering
is not unique.} Let us assume that the set of edges $E_H$ in the order of
enumeration is $\{e_1,\dots,e_{|E_H|}\}$. In the $i$-th round of GREEDY-COL,
edge $e_i$ is processed. GREEDY-COL involves exactly $|E_T|$ rounds of
coloring.\footnote{It may happen that in some rounds no rooted subtrees are
colored.}

\subsection{Coloring Strategy}
Let the set of rooted subtrees that are colored in the first $i$ rounds in
GREEDY-COL be $\mathcal{P}_i$. We define $\mathcal{P}_0$ to be empty. The set
of rooted subtrees present on edge $e_i$ but not in the set $\mathcal{P}_i$ is
denoted by $\mathcal{Q}_i$, i.e.,
$\mathcal{Q}_i=\mathcal{R}[e_i]\setminus\mathcal{P}_i$. Note that
$\mathcal{Q}_i$ is the set of rooted subtrees that are colored in the $i$-th
round of GREEDY-COL.

The basic idea is to be greedy in each round of coloring and try to use as few
new colors as possible while processing the edge.
The actual coloring scheme followed in the $i$-th round of GREEDY-COL depends
on the type of edge $e_i$ being processed. According to Lemma
\ref{lem:analysis1} below, tree edge $e_i$ encountered during the $i$-th round
of GREEDY-COL can be classified into one of the four types (defined in the
lemma) based on the status (whether already processed or not) of its adjacent
tree edges. If edge $e_i$ is of type $1$, $2$, or $3$ as defined in Lemma
\ref{lem:analysis1}, then uncolored rooted subtrees are randomly selected from
the set $\mathcal{Q}_i$ one at a time and are colored. In more detail, suppose
rooted subtree $\vec{R}$ has been selected from the set $\mathcal{Q}_i$ for
coloring. If there is a color that has already been assigned to some rooted
subtree(s) and can also be assigned to $\vec{R}$, then that color is assigned
to $\vec{R}$, otherwise a new color (not assigned to any other rooted subtree
previously) is assigned to $\vec{R}$.
On the other hand, if edge $e_i$ is of
type $4$ as defined in Lemma \ref{lem:analysis1}, then we assign colors to the
rooted subtrees in the set $\mathcal{Q}_i$ according to the better of the two
different coloring schemes presented as Subroutine \ref{algo:coloredge1}
(PROCESS-EDGE-1) and Subroutine \ref{algo:coloredge2} (PROCESS-EDGE-2).

\floatname{algorithm}{Subroutine}
\begin{footnotesize}
\begin{algorithm}
\footnotesize
\caption{PROCESS-EDGE-1}
\label{algo:coloredge1}
\begin{algorithmic}[1]
\REQUIRE
$\{T,\{u,v\}\in E_T,\mathcal{P},\mathcal{Q},\psi\}$ s.t. degree of tree
$T$ is at most $3$, $\mathcal{P}$ is set of rooted subtrees of $T$
that have already been colored according to
$\psi:\mathcal{P}\longrightarrow\mathbb{N}$ and $\mathcal{Q}$ is set of all
uncolored rooted subtrees of $T$ that are present on host tree edge
$\{u,v\}$.

\ENSURE Complete the given mapping $\psi$ to
$\psi:\mathcal{P}\cup\mathcal{Q}\longrightarrow\mathbb{N}$
s.t.
$\psi\in\Psi_{\mathcal{P}\cup\mathcal{Q}}$.

\STATE
$B_1\leftarrow G_{\mathcal{P}[\{u,v\}]\cup\mathcal{Q}}$\\

\FORALL{pairs
$\vec{R},\vec{S}\in\mathcal{P}[\{u,v\}]\cup\mathcal{Q}$ s.t.
$\vec{R},\vec{S}$ do not collide}\label{line:coloredge1_2}

\IF{any one of the following is true:
\begin{enumerate}
\item $\vec{R},\vec{S}\in\mathcal{P}$ and
$\psi(\vec{R})\neq\psi(\vec{S})$.
\item $\vec{R}\in\mathcal{Q},\vec{S}\in\mathcal{P}$ and
$\exists\,\vec{U}\in\mathcal{P}$ such that
$\psi(\vec{S})=\psi(\vec{U})$ and $\vec{R},\vec{U}$
collide.
\end{enumerate}}

\STATE $E_{B_1}\leftarrow E_{B_1}\cup\;\{\{\vec{R},\vec{S}\}\}$

\ENDIF

\ENDFOR\label{line:coloredge1_6}

\STATE Determine a maximum matching $M_{\bar{B}_1}\subseteq
E_{\bar{B}_1}$.\label{line:coloredge1_7}
\COMMENT{$\bar{B}_1$ is bipartite.}

\FORALL{matched edges $\{\vec{R},\vec{S}\}\in M_{\bar{B}_1}$
s.t. $\vec{R}\in\mathcal{Q}$ and $\vec{S}\in\mathcal{P}$}

\STATE $\psi(\vec{R})\leftarrow\psi(\vec{S})$\label{line:coloredge1_9}

\ENDFOR

\WHILE{$\exists$ some uncolored $\vec{R}\in\mathcal{Q}$}

\IF{$\exists$ matched edge
$\{\vec{R},\vec{S}\}\in M_{\bar{B}_1}$}

\STATE
$\psi(\vec{R}),\!\psi(\vec{S})\!\!\leftarrow\!\!\min\{m\!\in\!\mathbb{N}\!:\!\nexists\,\vec{U}\!\in\!\mathcal{P}\!\cup\!\mathcal{Q}
\textrm{ s.t. } \vec{R},\!\vec{U} \textrm{ or }
\vec{S},\!\vec{U} \textrm{ collide and } \psi(\vec{U})\!=\!m\}$\label{line:coloredge1_13}

\ELSE

\STATE
$\psi(\vec{R})\leftarrow\min\{m\in\mathbb{N}:\nexists\,\vec{U}\in\mathcal{P}\cup\mathcal{Q}
\textrm{ s.t. } \vec{R},\vec{U} \textrm{ collide and }
\psi(\vec{U})=m\}$

\ENDIF

\ENDWHILE

\end{algorithmic}
\end{algorithm}
\end{footnotesize}

As we shall see in Lemma \ref{lem:analysis1}, edge $e_i=\{u,v\}$ being a type
$4$ edge means that none of the tree edges adjacent to vertex $v$ have yet
been processed and there are two edges adjacent to vertex $u$ (besides edge
$e_i=\{u,v\}$), namely $\{u,w\}$ and $\{u,x\}$, of which edge $\{u,w\}$ has
already been processed and edge $\{u,x\}$ has not yet been processed. The two
coloring schemes employed while processing a type $4$ edge $e_i=\{u,v\}$
differ in the way they go about reusing the colors. In PROCESS-EDGE-1 we prefer
to reuse colors from the set
$\psi^{\mathrm{GDY}}(\mathcal{P}_{i-1}[\{u,v\}])$ (set of colors assigned to
the rooted subtree(s) present on host tree edge $e_i=\{u,v\}$ that were colored
in the first $i-1$ rounds), whereas in PROCESS-EDGE-2 we prefer to reuse colors
from the set
$\psi^{\mathrm{GDY}}(\mathcal{P}_{i-1}[\{u,x\}]\setminus\mathcal{P}_{i-1}[\{u,v\}])$
(set of colors assigned to the rooted subtree(s) present on host tree edge
$\{u,x\}$, but not on tree edge $e_i=\{u,v\}$, that were colored in the first
$i-1$ rounds). Note that the two sets of colors are not necessarily mutually
exclusive. The two schemes also differ in the order in which uncolored rooted
subtrees in the set $\mathcal{Q}_i$ are selected for coloring. More
specifically, in PROCESS-EDGE-2, first colors are assigned to all the rooted
subtrees in the set $\mathcal{Q}_i[\{u,x\}]$ and then to the rest of the
uncolored rooted subtrees.

\begin{footnotesize}
\begin{algorithm}
\footnotesize
\caption{PROCESS-EDGE-2}
\label{algo:coloredge2}
\begin{algorithmic}[1]
\REQUIRE
$\{T,\{\{u,v\},\{u,w\},\{u,x\}\}\subseteq E_T,\mathcal{P},\mathcal{Q},\psi\}$
s.t. degree of tree $T$ is $3$, $\mathcal{P}$ is set of rooted
subtrees of $T$ that have already been colored according to
$\psi:\mathcal{P}\longrightarrow\mathbb{N}$ and $\mathcal{Q}$ is set of all
uncolored rooted subtrees of $T$ that are present on host tree edge
$\{v,u\}$.

\ENSURE Complete the given mapping $\psi$ to
$\psi:\mathcal{P}\cup\mathcal{Q}\longrightarrow\mathbb{N}$ s.t.
$\psi\in\Psi_{\mathcal{P}\cup\mathcal{Q}}$.

\STATE
$B_2\leftarrow G_{\left(\mathcal{P}[\{u,x\}]\setminus\mathcal{P}[\{u,v\}]\right)\cup\mathcal{Q}[\{u,x\}]}$

\FORALL{pairs
$\vec{R},\vec{S}\!\in\!\left(\mathcal{P}[\{u,x\}]\!\setminus\!\mathcal{P}[\{u,v\}]\right)\cup\mathcal{Q}[\{u,x\}]$
s.t. $\vec{R},\vec{S}$ do not collide}

\IF{any one of the following is true:
\begin{enumerate}
\item $\vec{R},\vec{S}\in\mathcal{P}$ and
$\psi(\vec{R})\neq\psi(\vec{S})$.
\item $\vec{R}\in\mathcal{Q},\vec{S}\in\mathcal{P}$ and
$\exists\,\vec{U}\in\mathcal{P}$ s.t.
$\psi(\vec{S})=\psi(\vec{U})$ and $\vec{R},\vec{U}$
collide.
\end{enumerate}}

\STATE
$E_{B_2}\leftarrow E_{B_2}\cup\;\{\{\vec{R},\vec{S}\}\}$

\ENDIF

\ENDFOR

\STATE Determine a maximum matching
$M_{\bar{B}_2}\subseteq E_{\bar{B}_2}$.\label{line:coloredge2_7}
\COMMENT{$\bar{B}_2$ is bipartite.}

\FORALL{matched edges $\{\vec{R},\vec{S}\}\in M_{\bar{B}_2}$
s.t. $\vec{R}\in\mathcal{Q}$ and $\vec{S}\in\mathcal{P}$}

\STATE $\psi(\vec{R})\leftarrow\psi(\vec{S})$\label{line:coloredge2_9}

\ENDFOR

\WHILE{$\exists$ some uncolored $\vec{R}\in\mathcal{Q}[\{u,x\}]$}

\IF{$\exists$ matched edge
$\{\vec{R},\vec{S}\}\in M_{\bar{B}_2}$}

\STATE
$\psi(\vec{R}),\!\psi(\vec{S})\!\!\leftarrow\!\!\min\{m\!\in\!\mathbb{N}\!:\!\nexists\,\vec{U}\!\in\!\mathcal{P}\!\cup\!\mathcal{Q}
\textrm{ s.t. } \vec{R},\!\vec{U} \textrm{ or }
\vec{S},\!\vec{U} \textrm{ collide and } \psi(\vec{U})\!=\!m\}$\label{line:coloredge2_13}

\ELSE

\STATE
$\psi(\vec{R})\leftarrow\min\{m\in\mathbb{N}:\nexists\,\vec{U}\in\mathcal{P}\cup\mathcal{Q}
\textrm{ s.t. } \vec{R},\vec{U} \textrm{ collide and }
\psi(\vec{U})=m\}$\label{line:coloredge2_15}

\ENDIF

\ENDWHILE

\WHILE{$\exists$ some uncolored $\vec{R}\in\mathcal{Q}$}

\STATE
$\psi(\vec{R})\leftarrow\min\{m\in\mathbb{N}:\nexists\,\vec{U}\in\mathcal{P}\cup\mathcal{Q}
\textrm{ s.t. } \vec{R},\vec{U} \textrm{ collide and }
\psi(\vec{U})=m\}$\label{line:coloredge2_19}

\ENDWHILE

\end{algorithmic}
\end{algorithm}
\end{footnotesize}

In PROCESS-EDGE-1 (line \ref{line:coloredge1_7}), we determine the maximum
number of mutually exclusive pairs of rooted subtrees such that in each matched
pair (say $\vec{R},\vec{S}$) at least one of the rooted subtrees (say
$\vec{R}$) is an uncolored rooted subtree from the set $\mathcal{Q}_i$ (i.e.,
$\vec{R}\in\mathcal{Q}_i$) and the second rooted subtree ($\vec{S}$ in this
case) may either be (i) another uncolored rooted subtree from the set
$\mathcal{Q}_i$ (i.e., $\vec{S}\in\mathcal{Q}_i$) or (ii) a rooted subtree from
the set $\mathcal{P}_{i-1}[e_i]$ such that the uncolored rooted subtree in the
pair can be safely assigned its color (i.e., $\vec{S}\in\mathcal{P}_{i-1}$ such
that $\vec{R}$ does not collide with any rooted subtree that has already been
assigned the same color as $\vec{S}$). If the pair is of type (ii), then the
uncolored rooted subtree is assigned the same color as the other rooted subtree
(line \ref{line:coloredge1_9}). If the pair is of type (i), then both the
rooted subtrees of the pair are assigned the same color (line
\ref{line:coloredge1_13}). In this case, preference is given to the colors that
have already been assigned to some rooted subtree(s). If there is no such
suitable color, a new color is used.

In PROCESS-EDGE-2 (line \ref{line:coloredge2_7}), we determine the maximum
number of mutually exclusive pairs of rooted subtrees such that in each matched
pair (say $\vec{R},\vec{S}$) at least one of the rooted subtrees (say
$\vec{R}$) is an uncolored rooted subtree from the set $\mathcal{Q}_i$ and is
present on tree edge $\{u,x\}$ (i.e., $\vec{R}\in\mathcal{Q}_i[\{u,x\}]$) and
the second rooted subtree ($\vec{S}$ in this case) may either be (i) another
uncolored rooted subtree from the set $\mathcal{Q}_i$ present on edge $\{u,x\}$
(i.e., $\vec{S}\in\mathcal{Q}_i[\{u,x\}]$) or (ii) a rooted subtree from the
set $\mathcal{P}_{i-1}[\{u,x\}]\setminus\mathcal{P}_{i-1}[\{u,v\}]$ such that
the uncolored rooted subtree in the pair can be safely assigned its color (i.e.,
$\vec{S}\in\mathcal{P}_{i-1}[\{u,x\}]\setminus\mathcal{P}_{i-1}[\{u,v\}]$
such that $\vec{R}$ does not collide with any rooted subtree that has already
been assigned the same color as $\vec{S}$). If the pair is of type (ii), then
the uncolored rooted subtree is assigned the same color as the other rooted
subtree (line \ref{line:coloredge2_9}). If the pair is of type (i), then both
the rooted subtrees of the pair are assigned the same color (line
\ref{line:coloredge2_13}). Again preference is given to the colors that have
already been assigned to some rooted subtree(s). If there is no such suitable
color, a new color is used. After this all the remaining uncolored rooted
subtrees (all the rooted subtree in the set
$\mathcal{Q}_i\setminus\mathcal{Q}_i[\{u,x\}]$ and possibly some rooted
subtrees still uncolored in the set $\mathcal{Q}_i[\{u,x\}]$) are assigned
colors one at a time (lines \ref{line:coloredge2_15},
\ref{line:coloredge2_19}). Again preference is given to the colors that have
already been assigned to some rooted subtree(s).


\section{Analysis}
\label{sec:analysis}
In this section, we shall prove that GREEDY-COL is a
$\frac{5}{2}$-approximation algorithm for Problem \ref{prob:problem}.

\subsection{Some Local Properties}
We start off by stating a couple of straightforward but useful results about
the local structure of the problem at hand. Since the lemmas are very simple,
rather than complete proofs, we shall only give the intuition behind these.

\begin{lemma}
\label{lem:analysis7}
The complement of the conflict graph of any subset of rooted subtrees present
on a single host tree edge is bipartite.
\end{lemma}

\begin{lemma}
\label{lem:analysis0}
If the load of the set $\mathcal{R}$ of rooted subtrees on the tree $T$ is
$l_{\mathcal{R}}$ then there exists a set
$\mathcal{S}\supseteq\mathcal{R}$ of rooted subtrees on the tree $T$ such that
the following hold:
\begin{itemize}
\item The chromatic numbers of the conflict graphs $G_{\mathcal{S}}$ and
$G_{\mathcal{R}}$ are the same.
\item For every edge $\{u,v\}\in E_T$,
$|\mathcal{S}[(u,v)]|=|\mathcal{S}[(v,u)]|=l_{\mathcal{S}}=l_{\mathcal{R}}$.
\end{itemize}
Moreover, $\mathcal{S}$ can be constructed in polynomial time.
\end{lemma}

Lemma \ref{lem:analysis7} relies on two simple facts: (i) for any host tree
edge $\{u,v\}\in E_T$, $\mathcal{R}[(u,v)]$ and $\mathcal{R}[(v,u)]$ partition
$\mathcal{R}[\{u,v\}]$, and (ii) all rooted subtrees within each partition
collide with every other rooted subtree in that partition. The construction of
$\mathcal{S}$ in Lemma \ref{lem:analysis0} can be achieved by going over every
edge $\{u,v\}\in E_T$ and adding required number rooted subtrees containing
only the vertices $u$ and $v$ and one directed edge (either $(u,v)$ or $(v,u)$
such that $|\mathcal{S}[(u,v)]|$ and $|\mathcal{S}[(v,u)]|$ become equal to
$l_{\mathcal{R}}$. Lemma \ref{lem:analysis0} allows us to study only those
instances of the Problem \ref{prob:problem} where the given set of rooted
subtrees $\mathcal{R}$ is such that the number of rooted subtrees present on
all directed edges is the same. Going forward we shall assume that for every
edge $\{u,v\}\in E_T$,
$|\mathcal{R}[(u,v)]|=|\mathcal{R}[(v,u)]|=l_{\mathcal{R}}$. Consequently,
$|\mathcal{R}[\{u,v\}]|=2l_{\mathcal{R}}$.

\subsection{Roadmap}
First we give a roadmap that we shall follow for proving the approximation
ratio of $\frac{5}{2}$ for GREEDY-COL.
\begin{itemize}
\item First, in Lemmas \ref{lem:analysis1} and \ref{lem:analysis4}, we
characterize the types of host tree edges that are encountered during any round
of coloring in GREEDY-COL.
\item Next, in Lemma \ref{lem:analysis2}, we prove that if the edge to be
processed in $i$-th round of coloring is of type $1$, $2$ or $3$ as defined
in Lemma \ref{lem:analysis1}, then either no new colors are required in the
$i$-th round or the total number of colors in use at the end of the $i$-th
round is less than or equal to $2l_{\mathcal{R}}$.
\item If the edge to be processed in $i$-th round of coloring is of type $4$
as defined in Lemma \ref{lem:analysis1}, we first show in Lemma
\ref{lem:analysis8}, that either no new color is required in the $i$-th round
or the set of colors in use after $i$-th round of coloring is the same as the
set of colors used for the rooted subtrees that are assigned colors in the
$i$-th round ($\mathcal{Q}_i$) and all the rooted subtrees that are present on
host tree edge $\{u,w\}$ which is adjacent to the edge being processed in the
$i$-th round and has already been processed, i.e.,
$\psi^{\mathrm{GDY}}(\mathcal{Q}_i\cup\mathcal{P}_{i-1}[\{u,w\}])=\psi^{\mathrm{GDY}}(\mathcal{P}_i)$.
Next, we present bounds on the number of colors in the set
$\mathcal{Q}_i\cup\mathcal{P}_{i-1}[\{u,w\}]$ when subroutines
PROCESS-EDGE-1 (Lemma \ref{lem:analysis3}) or PROCESS-EDGE-2 (Lemma
\ref{lem:analysis5}) are employed.
\item Based on the previous lemmas, we determine the approximation ratio of
GREEDY-COL in a parameterized form in Lemma \ref{lem:analysis6}. In Lemma
\ref{lem:analysis9}, we determine the worst case (maximum) value of the
parameterized fraction obtained in Lemma \ref{lem:analysis6} to get the
approximation ratio.
\end{itemize}

\subsection{Host Tree Edge Types}
We characterize the host tree edge that is processed in any round of coloring
in GREEDY-COL by the status (whether already processed or not) of its adjacent
edges. Again, since the lemmas are straightforward, we only give the intuition
behind these.

\begin{lemma}
\label{lem:analysis1} In GREEDY-COL, when a host tree edge $\{u,v\}\in E_T$
(where $u$ was discovered before $v$ in the BFS) is being processed, then all
the edges adjacent to vertex $v$ are unprocessed, and for the edges adjacent to
vertex $u$ exactly one of the following is satisfied:
\begin{enumerate}
\item None of the edges adjacent to $u$ has been processed. In this case edge
$\{u,v\}$ is the first edge to be processed among all host tree edges.
\item Vertex $u$ has degree $2$ with adjacent edges
$\{u,v\},\{u,w\}$ of which edge $\{u,w\}$ has already been processed.
\item Vertex $u$ has degree $3$ with adjacent edges
$\{u,v\},\{u,w\},\{u,x\}$ of which edges $\{u,w\},\{u,x\}$ have already been processed.
\item Vertex $u$ has degree $3$ with adjacent edges
$\{u,v\},\{u,w\},\{u,x\}$ of which edge $\{u,w\}$ has already been processed while edge
$\{u,x\}$ has not yet been processed.
\end{enumerate}
\end{lemma}

Lemma \ref{lem:analysis1} relies on the fact that edges are processed in the
order of their discovery in a BFS, hence the set of edges that have been
processed during any time in GREEDY-COL form a tree. Using this observation, it
is simple to list the types of edges since we have restricted $T$ to have
degree at most $3$.


\begin{lemma}
\label{lem:analysis4}
In the $i$-th round of coloring in GREEDY-COL (while processing host tree edge
$e_i=\{u,v\}\in E_T$), if a rooted subtree $\vec{P}\in\mathcal{P}_{i-1}$, that
has already been assigned a color, collides with any uncolored rooted subtree
$\vec{Q}\in\mathcal{Q}_i$, then exactly one of the following is satisfied:
\begin{itemize}
\item Edge $e_i=\{u,v\}$ is of type $1$, $2$ or $3$ defined in Lemma
\ref{lem:analysis1}, and rooted subtree $\vec{P}\in\mathcal{P}_{i-1}[\{u,v\}]$.
\item Edge $e_i=\{u,v\}$ is of type $4$ defined in Lemma \ref{lem:analysis1},
and rooted subtree $\vec{P}\in\mathcal{P}_{i-1}[\{u,v\}]\cup\mathcal{P}_{i-1}[\{u,x\}]$
(where the vertex $x$ and the edge $\{u,x\}$ are as defined in Lemma
\ref{lem:analysis1}).
\end{itemize}
\end{lemma}

For Lemma \ref{lem:analysis4}, first observe that in case the edge
$\{u,v\}$ is of type $1$, $2$ or $3$ defined in Lemma \ref{lem:analysis1},
graph $G$ with vertex set $V_G=V_T$ and edge set $E_G=E_T\setminus\{u,v\}$ is a
forest of exactly $2$ trees say $T_u$ and $T_v$ (these may have only $1$
vertex). Moreover, since edges of $T$ are processed in a BFS based ordering,
the edges that have already been processed must all be present in one common
component of $G$. Let this be $T_u$. Hence all the rooted subtrees that have
already been colored must be present on $T_u$. On the other hand, again due to
the BFS based ordering, the rooted subtrees that are scheduled to be colored
while processing edge $\{u,v\}$ cannot be present on $T_u$ (otherwise they
would already have been colored). This implies that if a rooted subtree
$\vec{P}$, that was previously colored, collides with a rooted subtree that is
scheduled to be colored in this round of coloring, then the collision must
occur at some edge in $E_{T_v}\cup\{u,v\}$. But since $\vec{P}$ is present on
$T_u$, and the collision can only occur at some edge in $E_{T_v}\cup\{u,v\}$,
hence $\vec{P}$ must also be present on $\{u,v\}$. Proof for the case when
$\{u,v\}$ is of type $4$ defined in Lemma \ref{lem:analysis1} proceeds along
similar line of reasoning. The forest that we use for analysis in this case is
$G[V_T\setminus\{u\}]$.

\subsection{Type $1$, $2$, and $3$ Edges}
Next we give a bound on $\psi^{\mathrm{GDY}}(\mathcal{P}_i)$, the set of colors
used by GREEDY-COL for coloring to all the rooted subtrees present on host tree
edges processed in the first $i$ rounds of coloring, when the edge processed in
the $i$-th round of coloring is of type $1$, $2$ or $3$ defined in Lemma
\ref{lem:analysis1}.
\begin{lemma}
\label{lem:analysis2}
If edge $e_i=\{u,v\}$ being processed in the $i$-th round of GREEDY-COL is of
type $1$, $2$ or $3$ defined in Lemma \ref{lem:analysis1}, then
\begin{equation}
|\psi^{\mathrm{GDY}}(\mathcal{P}_i)|\leq\max\left\{2l_{\mathcal{R}},|\psi^{\mathrm{GDY}}(\mathcal{P}_{i-1})|\right\}.\nonumber
\end{equation}
\end{lemma}
\begin{proof}
First note that $\mathcal{R}[\{u,v\}]$, the set of rooted subtrees present on
host tree edge $e_i=\{u,v\}$, can be partitioned into $\mathcal{Q}_i$ and
$\mathcal{P}_{i-1}[\{u,v\}]$. Therefore
\begin{align}
\label{eq:analysis2}
|\mathcal{Q}_i|&=|\mathcal{R}[\{u,v\}]|-|\mathcal{P}_{i-1}[\{u,v\}]|\\
&\leq2l_{\mathcal{R}}-|\psi^{\mathrm{GDY}}(\mathcal{P}_{i-1}[\{u,v\}])|.\nonumber
\end{align}

Since the edge $e_i=\{u,v\}$ being processed in the $i$-th round of coloring is
of type $1$, $2$ or $3$ defined in Lemma \ref{lem:analysis1}, according to
Lemma \ref{lem:analysis4}, if a previously colored rooted subtree
$\vec{P}\in\mathcal{P}_{i-1}$ collides with any rooted subtree
$\vec{Q}\in\mathcal{Q}_i$ that is set to be colored in the $i$-th round, then
$\vec{P}\in\mathcal{P}_{i-1}[\{u,v\}]$. Hence, any color present in the set $\psi^{\mathrm{GDY}}(\mathcal{P}_{i-1})$ but absent in the set
$\psi^{\mathrm{GDY}}(\mathcal{P}_{i-1}[\{u,v\}])$ can be safely assigned to any
rooted subtree in the set $\mathcal{Q}_i$. There are $|\psi^{\mathrm{GDY}}(\mathcal{P}_{i-1})|-|\psi^{\mathrm{GDY}}(\mathcal{P}_{i-1}[\{u,v\}])|$
such colors. GREEDY-COL tries to reuse these colors first before using new
colors when coloring rooted subtrees in $\mathcal{Q}_i$. Therefore, the number
of new colors required in the $i$-th round is given by
\begin{align}
\label{eq:analysis3}
&\qquad|\psi^{\mathrm{GDY}}(\mathcal{P}_i)|-|\psi^{\mathrm{GDY}}(\mathcal{P}_{i-1})|\\
\leq&\bigg[|\mathcal{Q}_i|-\Big(|\psi^{\mathrm{GDY}}(\mathcal{P}_{i-1})|
-|\psi^{\mathrm{GDY}}(\mathcal{P}_{i-1}[\{u,v\}])|\Big)\bigg]^+\nonumber\\
\leq&\bigg[\Big(2l_{\mathcal{R}}-|\psi^{\mathrm{GDY}}(\mathcal{P}_{i-1}[\{u,v\}])|\Big)\nonumber\\
&\phantom{\bigg[}{}-\Big(|\psi^{\mathrm{GDY}}(\mathcal{P}_{i-1})|-|\psi^{\mathrm{GDY}}(\mathcal{P}_{i-1}[\{u,v\}])|\Big)\bigg]^+\nonumber\\
=&\left[2l_{\mathcal{R}}-|\psi^{\mathrm{GDY}}(\mathcal{P}_{i-1})|\right]^+,\nonumber
\end{align}
where the second inequality is by equation (\ref{eq:analysis2}). Equation
(\ref{eq:analysis3}) can now be rearranged to get the required result.
\end{proof}

\subsection{Type $4$ Edges}
Next we consider the case when edge $e_i=\{u,v\}$ being processed during the
$i$-th round of GREEDY-COL is of type $4$ defined in Lemma \ref{lem:analysis1}.
As stated in Lemma \ref{lem:analysis1}, we assume that edge $e_i=\{u,v\}$ is
such that (i) vertex $u$ was discovered before vertex $v$ in the BFS; (ii) all
the edges adjacent to vertex $v$ are unprocessed after the first $i-1$ rounds
of coloring; and (iii) vertex $u$ has degree $3$ with adjacent edges $\{u,v\}$,
$\{u,w\}$ and $\{u,x\}$ of which edge $\{u,w\}$ has already been processed
while edge $\{u,x\}$ has not yet been processed. In this case the set of
relevant rooted subtrees consist of $\mathcal{P}_{i-1}[\{u,w\}]$, the set of
rooted subtrees that have been assigned colors in the first $i-1$ rounds of
coloring and are present on the host tree edge $\{u,w\}$, and $\mathcal{Q}_i$,
the set of rooted subtrees that are to be colored in the $i$-th round. These
can be partitioned based on whether they are present or absent on the three
host tree edges $\{u,v\},\{u,w\},\{u,x\}$. This is shown in more detail in
Figure \ref{fig:analysis2}.

\begin{figure}
\centering
\includegraphics[width=5.5cm]{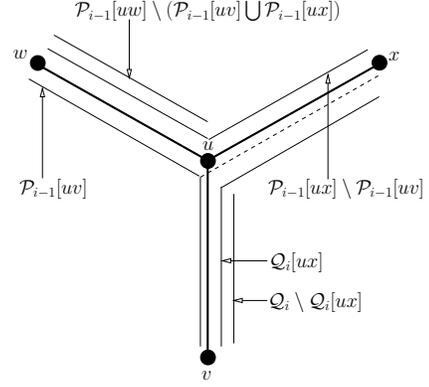}
\caption{\footnotesize{Sets of \emph{interesting} rooted subtrees encountered while
processing edge $\{u,v\}$ of type $4$ defined in Lemma \ref{lem:analysis1}.
Solid line on an edge implies that every rooted subtree of that set must be
present on that edge. Absence of a line on an edge implies that no rooted
subtree of that set can be present on that edge. Dotted line on an edge implies
that rooted subtrees of the set may or may not be present on that edge.}}
\label{fig:analysis2}
\end{figure}

\begin{lemma}
\label{lem:analysis8}
If edge $e_i=\{u,v\}$ being processed in the $i$-th round of GREEDY-COL is of
type $4$ defined in Lemma \ref{lem:analysis1}, then
\begin{align}
|\psi^{\mathrm{GDY}}(\mathcal{P}_i)|&=
\max\Big\{\big|\psi^{\mathrm{GDY}}(\mathcal{Q}_i\cup\mathcal{P}_{i-1}[\{u,w\}])\big|,\nonumber\\
&\phantom{=\max\Big\{}|\psi^{\mathrm{GDY}}(\mathcal{P}_{i-1})|\Big\},\nonumber
\end{align}
where the edge $\{u,w\}\in E_T$ is as defined in Lemma
\ref{lem:analysis1}.
\end{lemma}
\begin{proof}
Since the edge $e_i=\{u,v\}$ being processed in the $i$-th round of coloring is
of type $4$ defined in Lemma \ref{lem:analysis1}, according to Lemma
\ref{lem:analysis4}, if a rooted subtree that has already been colored in the
first $i-1$ rounds of GREEDY-COL collides with any rooted subtree that is to be
colored in the $i$-th round, then it must belong to the set
$\mathcal{P}_{i-1}[\{u,v\}]\cup\mathcal{P}_{i-1}[\{u,x\}]$. Since
$\mathcal{P}_{i-1}[\{u,v\}]\cup\mathcal{P}_{i-1}[\{u,x\}]\subseteq\mathcal{P}_{i-1}[\{u,w\}]$,
this implies that any rooted subtree in the set
$\mathcal{P}_{i-1}\setminus\mathcal{P}_{i-1}[\{u,w\}]$ cannot collide with any
rooted subtree in the set $\mathcal{Q}_i$. Therefore, any color already
assigned to some rooted subtree in the set
$\mathcal{P}_{i-1}\setminus\mathcal{P}_{i-1}[\{u,w\}]$, but not to any rooted
subtree in the set $\mathcal{P}_{i-1}[\{u,w\}]$, can be assigned to any rooted
subtree in the set $\mathcal{Q}_i$. There are
$|\psi^{\mathrm{GDY}}(\mathcal{P}_{i-1})|-|\psi^{\mathrm{GDY}}(\mathcal{P}_{i-1}[\{u,w\}])|$
such colors. During the $i$-th round of coloring, let
$\mathcal{N}_i\subseteq\mathcal{Q}_i$ be the set of rooted subtrees which do
not share colors with rooted subtrees in the set $\mathcal{P}_{i-1}[\{u,w\}]$,
i.e, $\mathcal{Q}_i\setminus\mathcal{N}_i$ is the largest subset of the set
$\mathcal{Q}_i$ such that
$|\psi^{\mathrm{GDY}}(\left(\mathcal{Q}_i\setminus\mathcal{N}_i\right)\cup\mathcal{P}_{i-1}[\{u,w\}])|=|\psi^{\mathrm{GDY}}(\mathcal{P}_{i-1}[\{u,w\}])|$.
We need $|\psi^{\mathrm{GDY}}(\mathcal{N}_i)|$ additional colors for coloring
all the rooted subtrees in the set $\mathcal{N}_i$ and there are
$|\psi^{\mathrm{GDY}}(\mathcal{P}_{i-1})|-|\psi^{\mathrm{GDY}}(\mathcal{P}_{i-1}[\{u,w\}])|$
available colors that can be used without adding any new color in the $i$-th
round of coloring. Therefore, the total number of colors required at the end of
$i$-th round of coloring in GREEDY-COL is
\begin{align}
&|\psi^{\mathrm{GDY}}(\mathcal{P}_i)|\nonumber\\
=&\bigg[|\psi^{\mathrm{GDY}}(\mathcal{N}_i)|-\Big(|\psi^{\mathrm{GDY}}(\mathcal{P}_{i-1})|\nonumber\\
&{}-|\psi^{\mathrm{GDY}}(\mathcal{P}_{i-1}[\{u,w\}])|\Big)\bigg]^++|\psi^{\mathrm{GDY}}(\mathcal{P}_{i-1})|\nonumber\\
=&\Big[|\psi^{\mathrm{GDY}}(\mathcal{N}_i)|+\big|\psi^{\mathrm{GDY}}(\left(\mathcal{Q}_i\setminus\mathcal{N}_i\right)\cup\mathcal{P}_{i-1}[\{u,w\}])\big|\nonumber\\
&{}-|\psi^{\mathrm{GDY}}(\mathcal{P}_{i-1})|\Big]^++|\psi^{\mathrm{GDY}}(\mathcal{P}_{i-1})|\nonumber\\
=&\Big[\big|\psi^{\mathrm{GDY}}(\mathcal{Q}_i\cup\mathcal{P}_{i-1}[\{u,w\}])\big|-|\psi^{\mathrm{GDY}}(\mathcal{P}_{i-1})|\Big]^+\nonumber\\
&{}+|\psi^{\mathrm{GDY}}(\mathcal{P}_{i-1})|\nonumber\\
=&\max\left\{\big|\psi^{\mathrm{GDY}}(\mathcal{Q}_i\cup\mathcal{P}_{i-1}[\{u,w\}])\big|,|\psi^{\mathrm{GDY}}(\mathcal{P}_{i-1})|\right\},\nonumber
\end{align}
where the third equality is due to the fact that the rooted subtrees
in the set $\mathcal{N}_i$ do not share any color with the rooted
subtrees in the set
$\left(\mathcal{Q}_i\setminus\mathcal{N}_i\right)\cup\mathcal{P}_{i-1}[\{u,w\}]$.
\end{proof}

Lemma \ref{lem:analysis8} suggests that we should develop bounds for
$|\psi^{\mathrm{GDY}}(\mathcal{Q}_i\cup\mathcal{P}_{i-1}[\{u,w\}])|$.
Using the notation of the lemma, if $\mathcal{N}_i\subseteq\mathcal{Q}_i$ is
the set of rooted subtrees that do not share colors with any rooted subtrees in
the set $\mathcal{P}_{i-1}[\{u,w\}]$, then
\begin{align}
&\big|\psi^{\mathrm{GDY}}(\mathcal{Q}_i\cup\mathcal{P}_{i-1}[\{u,w\}])\big|\nonumber\\
=&|\psi^{\mathrm{GDY}}(\mathcal{N}_i)|+\big|\psi^{\mathrm{GDY}}(\left(\mathcal{Q}_i\setminus\mathcal{N}_i\right)\cup\mathcal{P}_{i-1}[\{u,w\}])\big|\nonumber\\
=&|\psi^{\mathrm{GDY}}(\mathcal{N}_i)|+|\psi^{\mathrm{GDY}}(\mathcal{P}_{i-1}[\{u,w\}])|.\nonumber
\end{align}
Hence, in order to limit the use of new colors in the $i$-th round of coloring,
we try to minimize $|\psi^{\mathrm{GDY}}(\mathcal{N}_i)|$, the number of colors
used in the $i$-th round of coloring that are different from the colors
assigned to the rooted subtrees in the set $\mathcal{P}_{i-1}[\{u,w\}]$.

For any set $\mathcal{S}$ of rooted subtrees on the given host tree $T$ such
that the complement of their conflict graph is bipartite, i.e.,
$\bar{G}_{\mathcal{S}}$ is bipartite, we denote the size of maximum matching
\cite[p.67]{bollobas_book98} in $\bar{G}_{\mathcal{S}}$ by
$m_{\mathcal{S}}$.

\begin{lemma}
\label{lem:analysis3}
If the edge $e_i=\{u,v\}$ is of type $4$ defined in Lemma \ref{lem:analysis1},
and PROCESS-EDGE-1 is used for coloring in the $i$-th round of GREEDY-COL, then
\begin{align}
&\big|\psi^{\mathrm{GDY}}(\mathcal{Q}_i\cup\mathcal{P}_{i-1}[\{u,w\}])\big|\nonumber\\
\leq&2l_{\mathcal{R}}+|\mathcal{Q}_i|-m_{\mathcal{R}[\{u,v\}]}+m_{\mathcal{P}_{i-1}[\{u,v\}]}.\nonumber
\end{align}
\end{lemma}
\begin{proof}
In order to limit
$|\psi^{\mathrm{GDY}}(\mathcal{P}_{i-1}[\{u,w\}]\cup\mathcal{Q}_i)|-|\psi^{\mathrm{GDY}}(\mathcal{P}_{i-1}[\{u,w\}])|$,
PROCESS-EDGE-1 finds the maximum number of disjoint pairs $\vec{R},\vec{S}$ of
rooted subtrees such that one of the following is true:
\begin{itemize}
\item[(i)] Both $\vec{R},\vec{S}\in\mathcal{Q}_i$, and in this case they are
assigned the same (possibly new) color.
\item[(ii)] $\vec{R}\in\mathcal{Q}_i,\vec{S}\in\mathcal{P}_{i-1}[\{u,v\}]$, and in
this case $\vec{R}$ is assigned the same color as $\vec{S}$.
\end{itemize}
PROCESS-EDGE-1 finds such pairs of rooted subtrees by using graph $B_1$. Since
$\mathcal{P}_{i-1}[\{u,v\}]$ and $\mathcal{Q}_i$ partition the set
$\mathcal{R}[\{u,v\}]$, therefore by Lemma \ref{lem:analysis7}, graph
$\bar{G}_{\mathcal{P}_{i-1}[\{u,v\}]\cup\mathcal{Q}_i}$ is bipartite. This,
along with the fact that
$E_{G_{\mathcal{P}_{i-1}[\{u,v\}]\cup\mathcal{Q}_i}}\subseteq E_{B_1}$,
implies that $\bar{B}_1$ is also bipartite. Hence, it is easy to find a
maximum matching in $\bar{B}_1$. Let $M\subseteq E_{\bar{B}_1}$ be any matching
in $\bar{B}_1$. Observe that the edges are added to $B_1$ (lines
\ref{line:coloredge1_2}-\ref{line:coloredge1_6}) in such a way that if edge
$\{\vec{R},\vec{S}\}\in M$, then one of the following holds:
\begin{itemize}
\item[(i)] Both $\vec{R},\vec{S}\in\mathcal{Q}_i$.
\item[(ii)] $\vec{R}\in\mathcal{P}_{i-1},\vec{S}\in\mathcal{Q}_i$, and there is no
$\vec{U}\in\mathcal{P}_{i-1}$ such that $\vec{S},\vec{U}$ collide and
$\psi^{\mathrm{GDY}}(\vec{R})=\psi^{\mathrm{GDY}}(\vec{U})$.
\item[(iii)] Both $\vec{R},\vec{S}\in\mathcal{P}_{i-1}$ and
$\psi^{\mathrm{GDY}}(\vec{R})=\psi^{\mathrm{GDY}}(\vec{S})$.
\end{itemize}
So if edge $\{\vec{R},\vec{S}\}\in M$, then rooted subtrees $\vec{R}$ and
$\vec{S}$ can be assigned the same color. Note that the matched edges of type
(i) and (ii) correspond to the rooted subtree pairs of type (i) and (ii),
respectively. Matched edges of type (iii) simply list all the pairs of rooted
subtrees in the set $\mathcal{P}_{i-1}[\{u,v\}]$ that have already been
assigned the same colors. Since the number of edges of type (iii) is already
fixed, a maximum matching in $\bar{B}_1$ determines the maximum number of edges
of types (i) and (ii), i.e., it determines the maximum number of rooted subtree
pairs described above.

First assume that the rooted subtrees in the set $\mathcal{P}_{i-1}[\{u,v\}]$
do not share colors with any of the rooted subtree in the set
$\mathcal{P}_{i-1}[\{u,w\}]\setminus\mathcal{P}_{i-1}[\{u,v\}]$, although they
may share colors amongst themselves. As a consequence of Lemma
\ref{lem:analysis7}, more than two rooted subtrees in the set
$\mathcal{P}_{i-1}[\{u,v\}]$ cannot have the same color. Starting from any
maximum matching
$M_{\bar{G}_{\mathcal{P}_{i-1}[\{u,v\}]\cup\mathcal{Q}_i}}\subseteq E_{\bar{G}_{\mathcal{P}_{i-1}[\{u,v\}]\cup\mathcal{Q}_i}}$ in graph $\bar{G}_{\mathcal{P}_{i-1}[\{u,v\}]\cup\mathcal{Q}_i}$, we can construct a
matching $M\subseteq E_{\bar{B}_1}$ in graph $\bar{B}_1$ by first removing and
then adding the edges described next. We remove every matched edge
$\{\vec{R},\vec{S}\}\in M_{\bar{G}_{\mathcal{P}_{i-1}[\{u,v\}]\cup\mathcal{Q}_i}}$
for which one of the following is true:
\begin{itemize}
\item[(i)] Both $\vec{R},\vec{S}\in\mathcal{P}_{i-1}[\{u,v\}]$ such that
$\psi^{\mathrm{GDY}}(\vec{R})\neq\psi^{\mathrm{GDY}}(\vec{S})$, and there is no
rooted subtree $\vec{U}\in\mathcal{P}_{i-1}[\{u,v\}]$ such that
$\psi^{\mathrm{GDY}}(\vec{U})\notin\{\psi^{\mathrm{GDY}}(\vec{R}),\psi^{\mathrm{GDY}}(\vec{S})\}$.
\item[(ii)] Both $\vec{R},\vec{S}\in\mathcal{P}_{i-1}[\{u,v\}]$ such that
$\psi^{\mathrm{GDY}}(\vec{R})\neq\psi^{\mathrm{GDY}}(\vec{S})$, and there is a
rooted subtree $\vec{U}\in\mathcal{P}_{i-1}[\{u,v\}]$ such that
$\psi^{\mathrm{GDY}}(\vec{U})\in\{\psi^{\mathrm{GDY}}(\vec{R}),\psi^{\mathrm{GDY}}(\vec{S})\}$.
\item[(iii)] $\vec{R}\in\mathcal{Q}_i$, $\vec{S}\in\mathcal{P}_{i-1}[\{u,v\}]$,
and there is a rooted subtree $\vec{U}\in\mathcal{P}_{i-1}$ such that
$\psi^{\mathrm{GDY}}(\vec{U})=\psi^{\mathrm{GDY}}(\vec{S})$.
\end{itemize}
Consider rooted subtrees $\vec{R},\vec{S}\in\mathcal{P}_{i-1}[\{u,v\}]$ with
$\psi^{\mathrm{GDY}}(\vec{R})=\psi^{\mathrm{GDY}}(\vec{S})$. Since $M_{\bar{G}_{\mathcal{P}_{i-1}[\{u,v\}]\cup\mathcal{Q}_i}}$ is a maximum
matching in $\bar{G}_{\mathcal{P}_{i-1}[\{u,v\}]\cup\mathcal{Q}_i}$, either
edge $\{\vec{R},\vec{S}\}\in M_{\bar{G}_{\mathcal{P}_{i-1}[\{u,v\}]\cup\mathcal{Q}_i}}$,
or at least one of the rooted subtrees $\vec{R},\vec{S}$ is matched to some
other rooted subtree in
$M_{\bar{G}_{\mathcal{P}_{i-1}[\{u,v\}]\cup\mathcal{Q}_i}}$.\footnote{It may
happen that both the rooted subtrees $\vec{R},\vec{S}$ are matched to different
vertices in $M_{\bar{G}_{\mathcal{P}_{i-1}[\{u,v\}]\cup\mathcal{Q}_i}}$} In the
case when rooted subtrees $\vec{R},\vec{S}$ are not already matched to each
other in $M_{\bar{G}_{\mathcal{P}_{i-1}[\{u,v\}]\cup\mathcal{Q}_i}}$, the
edge(s) adjacent to $\vec{R}$ or $\vec{S}$ (or both) in
$M_{\bar{G}_{\mathcal{P}_{i-1}[\{u,v\}]\cup\mathcal{Q}_i}}$ is (are) either of
type (ii) or of type (iii) and is (are) therefore removed from the matching.
Hence, we can safely add edge $\{\vec{R},\vec{S}\}$ to the matching. Let the
set of removed edges of type (i), (ii) and (iii) be $E_{\mathrm{r(i)}}$,
$E_{\mathrm{r(ii)}}$ and $E_{\mathrm{r(iii)}}$, respectively, and the set of
added edges be $E_{\mathrm{a}}$. Observe that for every removed edge in the set
$E_{\mathrm{r(ii)}}\cup E_{\mathrm{r(iii)}}$, there is a corresponding edge in
the set $E_{\mathrm{a}}$ added to the matching such that for at most two
removed edges in the set $E_{\mathrm{r(ii)}}\cup E_{\mathrm{r(iii)}}$, the
corresponding added edge in the set $E_{\mathrm{a}}$ can be the same; therefore
$|E_{\mathrm{a}}|\geq\frac{1}{2}(|E_{\mathrm{r(ii)}}|+|E_{\mathrm{r(iii)}}|)$.
We can now lower bound the size of maximum matching
$M_{\bar{B}_1}\subseteq E_{\bar{B}_1}$ in graph $\bar{B}_1$ by the size of $M$,
a valid matching in the graph. $|M|$ is equal to
$|M_{\bar{G}_{\mathcal{P}_{i-1}[\{u,v\}]\cup\mathcal{Q}_i}}|=m_{\{T\mathcal{P}_{i-1}[\{u,v\}]\cup\mathcal{Q}_i\}}$
minus the number of edges removed plus the number of edges added. Thus
\begin{align}
\label{eq:ineq3a}
&|M_{\bar{B}_1}|\geq|M|\\
=&m_{\mathcal{P}_{i-1}[\{u,v\}]\cup\mathcal{Q}_i}-\left(|E_{\mathrm{r(i)}}|+|E_{\mathrm{r(ii)}}|+|E_{\mathrm{r(iii)}}|-|E_{\mathrm{a}}|\right)\nonumber\\
\geq&m_{\mathcal{P}_{i-1}[\{u,v\}]\cup\mathcal{Q}_i}-\left(|E_{\mathrm{r(i)}}|+|E_{\mathrm{a}}|\right)\nonumber\\
\geq&m_{\mathcal{P}_{i-1}[\{u,v\}]\cup\mathcal{Q}_i}-m_{\mathcal{P}_{i-1}[\{u,v\}]},\nonumber
\end{align}
where we are using the fact that $E_{\mathrm{a}}\cup E_{\mathrm{r(i)}}$, the
set of removed edges of type (i) and the set of added edges form a matching in
the bipartite graph $\bar{G}_{\mathcal{P}_{i-1}[\{u,v\}]}$. This is because
$E_{\mathrm{a}}\cup E_{\mathrm{r(i)}}\subseteq E_{\bar{G}_{\mathcal{P}_{i-1}[\{u,v\}]}}$,
and the end vertices of edges in the sets $E_{\mathrm{a}},E_{\mathrm{r(i)}}$
are distinct.

The vertex set $V_{\bar{B}_1}$ corresponds to all the rooted subtrees in the
set $\mathcal{P}_{i-1}[\{u,v\}]\cup\mathcal{Q}_i$, and an edge in matching
$M_{\bar{B}_1}$ determines two rooted subtrees which share their color after
this round of coloring. Therefore, using inequality (\ref{eq:ineq3a}) and the
fact that the subsets $\mathcal{P}_{i-1}[\{u,v\}]$ and $\mathcal{Q}_i$ partition
the set $\mathcal{R}[\{u,v\}]$,
\begin{align}
\label{eq:ineq3c}
&\big|\psi^{\mathrm{GDY}}(\mathcal{P}_{i-1}[\{u,v\}]\cup\mathcal{Q}_i)\big|\\
\leq&\big|\mathcal{P}_{i-1}[\{u,v\}]\cup\mathcal{Q}_i\big|-|M_{\bar{B}_1}|\nonumber\\
\leq&|\mathcal{P}_{i-1}[\{u,v\}]|+|\mathcal{Q}_i|-m_{\mathcal{P}_{i-1}[\{u,v\}]\cup\mathcal{Q}_i}
+m_{\mathcal{P}_{i-1}[\{u,v\}]}.\nonumber
\end{align}
Using inequality (\ref{eq:ineq3c})
\begin{align}
\label{eq:ineq3b}
&\big|\psi^{\mathrm{GDY}}(\mathcal{P}_{i-1}[\{u,w\}]\cup\mathcal{Q}_i)\big|\\
=&\big|\psi^{\mathrm{GDY}}(\mathcal{P}_{i-1}[\{u,w\}]\cup\mathcal{Q}_i)\big|-\big|\psi^{\mathrm{GDY}}(\mathcal{P}_{i-1}[\{u,v\}]\cup\mathcal{Q}_i)\big|\nonumber\\
&{}+\big|\psi^{\mathrm{GDY}}(\mathcal{P}_{i-1}[\{u,v\}]\cup\mathcal{Q}_i)\big|\nonumber\\
\leq&|\mathcal{P}_{i-1}[\{u,w\}]\setminus\mathcal{P}_{i-1}[\{u,v\}]|+|\mathcal{P}_{i-1}[\{u,v\}]|+|\mathcal{Q}_i|\nonumber\\
&{}-m_{\mathcal{P}_{i-1}[\{u,v\}]\cup\mathcal{Q}_i}+m_{\mathcal{P}_{i-1}[\{u,v\}]}\nonumber\\
\leq&2l_{\mathcal{R}}+|\mathcal{Q}_i|-m_{\mathcal{R}[\{u,v\}]}+m_{\mathcal{P}_{i-1}[\{u,v\}]}.\nonumber
\end{align}
The first inequality uses the fact that
$|\psi^{\mathrm{GDY}}(\mathcal{P}_{i-1}[\{u,w\}]\cup\mathcal{Q}_i)|-|\psi^{\mathrm{GDY}}(\mathcal{P}_{i-1}[\{u,v\}]\cup\mathcal{Q}_i)|$
is the number of colors used for coloring all the rooted subtrees in the set
$\mathcal{P}_{i-1}[\{u,w\}]\setminus\mathcal{P}_{i-1}[\{u,v\}]$ that are
different from the colors used for coloring rooted subtrees in the set
$\mathcal{P}_{i-1}[\{u,v\}]\cup\mathcal{Q}_i$; therefore, this number is upper
bounded by $|\mathcal{P}_{i-1}[\{u,w\}]\setminus\mathcal{P}_{i-1}[\{u,v\}]|$.
The final inequality uses the fact that the subsets
$\mathcal{P}_{i-1}[\{u,v\}]$ and
$\mathcal{P}_{i-1}[\{u,w\}]\setminus\mathcal{P}_{i-1}[\{u,v\}]$
partition the set $\mathcal{P}_{i-1}[\{u,w\}]=\mathcal{R}[\{u,w\}]$.

Next, suppose some rooted subtree $\vec{R}\in\mathcal{P}_{i-1}[\{u,v\}]$ shares
its color with another rooted subtree
$\vec{S}\in\mathcal{P}_{i-1}[\{u,w\}]\setminus\mathcal{P}_{i-1}[\{u,v\}]$.
In this case, the worst that can happen is that some rooted subtrees in the set
$\mathcal{Q}_i$, that could have shared color with rooted subtree $\vec{R}$,
can no longer do so since they collide with rooted subtree $\vec{S}$. Hence the
size of maximum matching $M_{\bar{B}_1}$ reduces by $1$. The unit reduction is
independent of the number of affected rooted subtrees in the set
$\mathcal{Q}_i$, since in $M_{\bar{B}_1}$ rooted subtree $\vec{R}$ can be
potentially matched to only one of them. On the other hand, the rooted subtrees
$\vec{R}\in\mathcal{P}_{i-1}[\{u,v\}],\vec{S}\in\mathcal{P}_{i-1}[\{u,w\}]\setminus\mathcal{P}_{i-1}[\{u,v\}]$
sharing color means that
$|\psi^{\mathrm{GDY}}(\mathcal{P}_{i-1}[\{u,w\}]\cup\mathcal{Q}_i)|-|\psi^{\mathrm{GDY}}(\mathcal{P}_{i-1}[\{u,v\}]\cup\mathcal{Q}_i)|$,
the number of colors used for assigning colors to all the rooted subtrees in
the set $\mathcal{P}_{i-1}[\{u,w\}]\setminus\mathcal{P}_{i-1}[\{u,v\}]$ that
are different from the colors used for assigning colors to the rooted subtrees
in the set $\mathcal{P}_{i-1}[\{u,v\}]\cup\mathcal{Q}_i$, also reduces by $1$.
Applying both the observations, we note that the final inequality in
(\ref{eq:ineq3b}) still holds.
\end{proof}

\begin{lemma}
\label{lem:analysis5}
If the edge $e_i=\{u,v\}$ is of type $4$ defined in Lemma \ref{lem:analysis1},
and PROCESS-EDGE-2 is used for coloring in the $i$-th round of GREEDY-COL, then
\begin{equation}
\big|\psi^{\mathrm{GDY}}(\mathcal{P}_{i-1}[\{u,w\}]\cup\mathcal{Q}_i)\big|\leq2l_{\{\vec{T}_H,\mathcal{R}\}}+\left[g-h\right]^+,\nonumber
\end{equation}
where
\begin{align}
g=&|\mathcal{Q}_i[\{u,x\}]|+|\mathcal{P}_{i-1}[\{u,x\}]\setminus\mathcal{P}_{i-1}[\{u,v\}]|-|\mathcal{Q}_i|,\nonumber\\
h=&\Bigg[|\mathcal{Q}_i[\{u,x\}]|+\frac{|\mathcal{P}_{i-1}[\{u,x\}]\setminus\mathcal{P}_{i-1}[\{u,v\}]|}{2}\nonumber\\
&{}+m_{\mathcal{R}[\{u,x\}]}-2l_{\mathcal{R}}\Bigg]^+.\nonumber
\end{align}
\end{lemma}
\begin{proof}
Observe that $\mathcal{P}_{i-1}[\{u,v\}]$ and
$\mathcal{P}_{i-1}[\{u,w\}]\setminus\mathcal{P}_{i-1}[\{u,v\}]$ partition $\mathcal{P}_{i-1}[\{u,w\}]=\mathcal{R}[\{u,w\}]$, and
$\mathcal{P}_{i-1}[\{u,v\}]$ and $\mathcal{Q}_i$ partition
$\mathcal{R}[\{u,v\}]$. Since
$|\mathcal{R}[\{u,w\}]|=|\mathcal{R}[\{u,v\}]|=2l_{\mathcal{R}}$, we have
\begin{equation}
\label{eq:eq3}
|\mathcal{P}_{i-1}[\{u,w\}]\setminus\mathcal{P}_{i-1}[\{u,v\}]|=|\mathcal{Q}_i|.
\end{equation}
Since $\mathcal{Q}_i[\{u,x\}]$ and $\mathcal{Q}_i\setminus\mathcal{Q}_i[\{u,x\}]$
partition $\mathcal{Q}_i$, and
$\mathcal{P}_{i-1}[\{u,x\}]\setminus\mathcal{P}_{i-1}[\{u,v\}]$ and
$\mathcal{P}_{i-1}[\{u,w\}]\setminus\left(\mathcal{P}_{i-1}[\{u,v\}]\cup\mathcal{P}_{i-1}[\{u,x\}]\right)$
partition $\mathcal{P}_{i-1}[\{u,w\}]\setminus\mathcal{P}_{i-1}[\{u,v\}]$,
from equation (\ref{eq:eq3})
\begin{align}
\label{eq:four}
&\big|\mathcal{P}_{i-1}[\{u,w\}]\setminus\left(\mathcal{P}_{i-1}[\{u,v\}]\cup\mathcal{P}_{i-1}[\{u,x\}]\right)\big|\\
&{}+|\mathcal{P}_{i-1}[\{u,x\}]\setminus\mathcal{P}_{i-1}[\{u,v\}]|\nonumber\\
=&|\mathcal{P}_{i-1}[\{u,w\}]\setminus\mathcal{P}_{i-1}[\{u,v\}]|\nonumber\\
=&|\mathcal{Q}_i\setminus\mathcal{Q}_i[\{u,x\}]|+|\mathcal{Q}_i[\{u,x\}]|=|\mathcal{Q}_i|.\nonumber
\end{align}
PROCESS-EDGE-2 first finds the maximum number of disjoint pairs
$\vec{R},\vec{S}$ of rooted subtrees such that one of the following is true:
\begin{itemize}
\item[(i)] Both $\vec{R},\vec{S}\in\mathcal{Q}_i[\{u,x\}]$. In this case, both
$\vec{R}$ and $\vec{S}$ are assigned the same color (we shall specify exactly
which color is assigned in a moment).
\item[(ii)] $\vec{R}\in\mathcal{Q}_i[\{u,x\}]$ and
$\vec{S}\in\mathcal{P}_{i-1}[\{u,x\}]\setminus\mathcal{P}_{i-1}[\{u,v\}]$
such that $\vec{R}$ can be assigned the same color as $\vec{S}$. In this case
$\vec{R}$ is indeed assigned the same color as $\vec{S}$.
\end{itemize}
PROCESS-EDGE-2 finds such pairs of rooted subtrees by using graph $B_2$. Since
$\mathcal{P}_{i-1}[\{u,x\}]\setminus\mathcal{P}_{i-1}[\{u,v\}]$ and
$\mathcal{Q}_i[\{u,x\}]$ are disjoint subsets of $\mathcal{R}[\{u,x\}]$, by
Lemma \ref{lem:analysis7} the graph
$\bar{G}_{(\mathcal{P}_{i-1}[\{u,x\}]\setminus\mathcal{P}_{i-1}[\{u,v\}])\cup\mathcal{Q}_i[\{u,x\}]}$
is bipartite. This, along with the fact that
$E_{G_{(\mathcal{P}_{i-1}[\{u,x\}]\setminus\mathcal{P}_{i-1}[\{u,v\}])\cup\mathcal{Q}_i[\{u,x\}]}}\subseteq E_{B_2}$, implies that $\bar{B}_2$ is also bipartite. Hence, it is easy to find
a maximum matching in $\bar{B}_2$. Let $M\subseteq E_{\bar{B}_2}$ be any
matching in $\bar{B}_2$. Edges are added to $B_2$ in such a way that if edge
$\{\vec{R},\vec{S}\}\in M$, then one of the following holds:
\begin{itemize}
\item[(i)] Both $\vec{R},\vec{S}\in\mathcal{Q}_i[\{u,x\}]$.
\item[(ii)] $\vec{R}\in\mathcal{Q}_i[\{u,x\}],\vec{S}\in\mathcal{P}_{i-1}[\{u,x\}]\setminus\mathcal{P}_{i-1}[\{u,v\}]$,
and there is no $\vec{U}\in\mathcal{P}_{i-1}$ such that $\vec{R},\vec{U}$
collide and $\psi^{\mathrm{GDY}}(\vec{S})=\psi^{\mathrm{GDY}}(\vec{U})$.
\item[(iii)] Both
$\vec{R},\vec{S}\in\mathcal{P}_{i-1}[\{u,x\}]\setminus\mathcal{P}_{i-1}[\{u,v\}]$
and $\psi^{\mathrm{GDY}}(\vec{R})=\psi^{\mathrm{GDY}}(\vec{S})$.
\end{itemize}
So if edge $\{\vec{R},\vec{S}\}\in M$, then the rooted subtrees
$\vec{R},\vec{S}$ can be assigned the same color. Note that the matched edges
of type (i) and (ii) correspond to the rooted subtree pairs of type (i) and (ii),
respectively. Matched edges of type (iii) simply list all the pairs of rooted
subtrees in the set
$\mathcal{P}_{i-1}[\{u,x\}]\setminus\mathcal{P}_{i-1}[\{u,v\}]$ that have
already been colored. Since the number of edges of type (iii) is already fixed,
a maximum matching in $\bar{B}_2$ determines the maximum number of edges of
types (i) and (ii), i.e., it determines the maximum number of rooted subtree
pairs described above.

First, we assume that the rooted subtrees in the set
$\mathcal{P}_{i-1}[\{u,w\}]\setminus\mathcal{P}_{i-1}[\{u,v\}]$ do not share
colors with any rooted subtree in the set $\mathcal{P}_{i-1}[\{u,v\}]$,
although they may share colors amongst themselves. Let
$M_{\bar{B}_2}\subseteq E_{\bar{B}_2}$ be a maximum matching in $\bar{B}_2$.
Let the number of type (i), (ii) and (iii) edges in the matching be
$t_1,t_2,t_3$, respectively. In this case the size of the maximum matching in
$\bar{B}_2$ is lower bounded as
\begin{align}
\label{eq:one}
&|M_{\bar{B}_2}|=\sum_{j=1}^{3}t_j\\
\geq&m_{\mathcal{Q}_i[\{u,x\}]\cup\left(\mathcal{P}_{i-1}[\{u,x\}]\setminus\mathcal{P}_{i-1}[\{u,v\}]\right)}\nonumber\\
&{}-m_{\mathcal{P}_{i-1}[\{u,x\}]\setminus\mathcal{P}_{i-1}[\{u,v\}]}\nonumber\\
\geq&\bigg[m_{\mathcal{Q}_i[\{u,x\}]\cup\left(\mathcal{P}_{i-1}[\{u,x\}]\setminus\mathcal{P}_{i-1}[\{u,v\}]\right)}\nonumber\\
&{}-\frac{|\mathcal{P}_{i-1}[\{u,x\}]\setminus\mathcal{P}_{i-1}[\{u,v\}]|}{2}\bigg]^+,\nonumber
\end{align}
where
$m_{\mathcal{Q}_i[\{u,x\}]\cup\left(\mathcal{P}_{i-1}[\{u,x\}]\setminus\mathcal{P}_{i-1}[\{u,v\}]\right)}$
and $m_{\mathcal{P}_{i-1}[\{u,x\}]\setminus\mathcal{P}_{i-1}[\{u,v\}]}$
are the sizes of maximum matchings in the bipartite graphs
$\bar{G}_{\mathcal{Q}_i[\{u,x\}]\cup\left(\mathcal{P}_{i-1}[\{u,x\}]\setminus\mathcal{P}_{i-1}[\{u,v\}]\right)}$
and $\bar{G}_{\mathcal{P}_{i-1}[\{u,x\}]\setminus\mathcal{P}_{i-1}[\{u,v\}]}$,
respectively. The reasoning for the initial inequality follows exactly as the
reasoning for inequality (\ref{eq:ineq3a}) presented in the proof of Lemma
\ref{lem:analysis3}. For the final inequality, we use the facts that the size
of any matching in the bipartite graph
$\bar{G}_{\mathcal{P}_{i-1}[\{u,x\}]\setminus\mathcal{P}_{i-1}[\{u,v\}]}$
must be smaller than half of the size of its vertex set, and the size of a
matching cannot be negative. Since
$\bar{G}_{\mathcal{Q}_i[\{u,x\}]\cup\left(\mathcal{P}_{i-1}[\{u,x\}]\setminus\mathcal{P}_{i-1}[\{u,v\}]\right)}$
is a subgraph of $\bar{G}_{\mathcal{R}[\{u,x\}]}$ induced by the vertex set
corresponding to the rooted subtrees in the set
$\mathcal{Q}_i[\{u,x\}]\cup\left(\mathcal{P}_{i-1}[\{u,x\}]\setminus\mathcal{P}_{i-1}[\{u,v\}]\right)$,
if the size of a maximum matching in $\bar{G}_{\mathcal{R}[\{u,x\}]}$ is
$m_{\mathcal{R}[\{u,x\}]}$, then the size of a maximum matching in
$\bar{G}_{\mathcal{Q}_i[\{u,x\}]\cup\left(\mathcal{P}_{i-1}[\{u,x\}]\setminus\mathcal{P}_{i-1}[\{u,v\}]\right)}$
is bounded as
\begin{align}
\label{eq:two}
&m_{\mathcal{Q}_i[\{u,x\}]\cup\left(\mathcal{P}_{i-1}[\{u,x\}]\setminus\mathcal{P}_{i-1}[\{u,v\}]\right)}\\
\geq&\bigg[m_{\mathcal{R}[\{u,x\}]}-|\mathcal{R}[\{u,x\}]|+\big|\mathcal{Q}_i[\{u,x\}]\cup\nonumber\\
&\Big(\mathcal{P}_{i-1}[\{u,x\}]\setminus\mathcal{P}_{i-1}[\{u,v\}]\Big)\big|\bigg]^+\nonumber\\
=&\bigg[|\mathcal{Q}_i[\{u,x\}]|+|\mathcal{P}_{i-1}[\{u,x\}]\setminus\mathcal{P}_{i-1}[\{u,v\}]|\nonumber\\
&{}+m_{\mathcal{R}[\{u,x\}]}-2l_{\mathcal{R}}\bigg]^+.\nonumber
\end{align}
This is because if we consider a maximum matching
$M_{\bar{G}_{\mathcal{R}[\{u,x\}]}}\subseteq E_{\bar{G}_{\mathcal{R}[\{u,x\}]}}$
in the graph $\bar{G}_{\mathcal{R}[\{u,x\}]}$, any edge
$\{\vec{R},\vec{S}\}\in M_{\bar{G}_{\mathcal{R}[\{u,x\}]}}$ can be classified
into one of the following three types:
\begin{itemize}
\item[(i)] Both $\vec{R},\vec{S}\in\mathcal{Q}_i[\{u,x\}]\cup\left(\mathcal{P}_{i-1}[\{u,x\}]\setminus\mathcal{P}_{i-1}[\{u,v\}]\right)$.
\item[(ii)] Rooted subtree $\vec{R}\in\mathcal{Q}_i[\{u,x\}]\cup\left(\mathcal{P}_{i-1}[\{u,x\}]\setminus\mathcal{P}_{i-1}[\{u,v\}]\right)$
whereas rooted subtree
$\vec{S}\in\mathcal{R}[\{u,x\}]\setminus\left(\mathcal{Q}_i[\{u,x\}]\cup\left(\mathcal{P}_{i-1}[\{u,x\}]\setminus\mathcal{P}_{i-1}[\{u,v\}]\right)\right)$.
\item[(iii)] Both $\vec{R},\vec{S}\in\mathcal{R}[\{u,x\}]\setminus\left(\mathcal{Q}_i[\{u,x\}]\cup\left(\mathcal{P}_{i-1}[\{u,x\}]\setminus\mathcal{P}_{i-1}[\{u,v\}]\right)\right)$.
\end{itemize}
Let the set of edges of type (i), (ii) and (iii) be
$E_{\mathrm{(i)}},E_{\mathrm{(ii)}},E_{\mathrm{(iii)}}$,
respectively. Clearly, $E_{\mathrm{(i)}}$ is a valid matching in the graph
$\bar{G}_{\mathcal{Q}_i[\{u,x\}]\cup\left(\mathcal{P}_{i-1}[\{u,x\}]\setminus\mathcal{P}_{i-1}[\{u,v\}]\right)}$,
therefore a lower bound for $|E_{\mathrm{(i)}}|$ can be treated as
a lower bound for
$m_{\mathcal{Q}_i[\{u,x\}]\cup\left(\mathcal{P}_{i-1}[\{u,x\}]\setminus\mathcal{P}_{i-1}[\{u,v\}]\right)}$.
Also, since maximum matching $M_{\bar{G}_{\mathcal{R}[\{u,x\}]}}$ can be
partitioned into sets $E_{\mathrm{(i)}},E_{\mathrm{(ii)}},E_{\mathrm{(iii)}}$,
we get
\begin{align}
\label{eq:twoa}
&m_{\mathcal{Q}_i[\{u,x\}]\cup\left(\mathcal{P}_{i-1}[\{u,x\}]\setminus\mathcal{P}_{i-1}[\{u,v\}]\right)}\\
\geq&|E_{\mathrm{(i)}}|\geq m_{\mathcal{R}[\{u,x\}]}-|E_{\mathrm{(ii)}}|-|E_{\mathrm{(iii)}}|.\nonumber
\end{align}
Since an edge in the set $E_{\mathrm{(ii)}}$ requires one of the rooted
subtree, and an edge in the set $E_{\mathrm{(iii)}}$ requires both of the
rooted subtrees to be from the set
$\mathcal{R}[\{u,x\}]\setminus\left(\mathcal{Q}_i[\{u,x\}]\cup\left(\mathcal{P}_{i-1}[\{u,x\}]\setminus\mathcal{P}_{i-1}[\{u,v\}]\right)\right)$,
we have
\begin{align}
\label{eq:twob}
&\qquad\qquad|E_{\mathrm{(ii)}}|+2|E_{\mathrm{(iii)}}|\\
\leq&\big|\mathcal{R}[\{u,x\}]\setminus\left(\mathcal{Q}_i[\{u,x\}]\cup\left(\mathcal{P}_{i-1}[\{u,x\}]\setminus\mathcal{P}_{i-1}[\{u,v\}]\right)\right)\big|\nonumber\\
=&|\mathcal{R}[\{u,x\}]|-\big|\mathcal{Q}_i[\{u,x\}]\cup\left(\mathcal{P}_{i-1}[\{u,x\}]\setminus\mathcal{P}_{i-1}[\{u,v\}]\right)\big|.\nonumber
\end{align}
From inequalities (\ref{eq:twoa}), (\ref{eq:twob}) and the fact that the size
of a matching cannot be negative, we obtain the required inequality
(\ref{eq:two}). From equations (\ref{eq:one}) and (\ref{eq:two}),
\begin{align}
\label{eq:three}
&\qquad\qquad|M_{\bar{B}_2}|=\sum_{j=1}^{3}t_j\\
\geq&\bigg[\Big[|\mathcal{Q}_i[\{u,x\}]|+|\mathcal{P}_{i-1}[\{u,x\}]\setminus\mathcal{P}_{i-1}[\{u,v\}]|\nonumber\\
&{}+m_{\mathcal{R}[\{u,x\}]}-2l_{\mathcal{R}}\Big]^+\nonumber\\
&{}-\frac{|\mathcal{P}_{i-1}[\{u,x\}]\setminus\mathcal{P}_{i-1}[\{u,v\}]|}{2}\bigg]^+\nonumber\\
=&\Bigg[|\mathcal{Q}_i[\{u,x\}]|+\frac{|\mathcal{P}_{i-1}[\{u,x\}]\setminus\mathcal{P}_{i-1}[\{u,v\}]|}{2}\nonumber\\
&{}+m_{\mathcal{R}[\{u,x\}]}-2l_{\mathcal{R}}\Bigg]^+=h.\nonumber
\end{align}

PROCESS-EDGE-2 assigns colors to the uncolored rooted subtrees in the set
$\mathcal{Q}_i$ in the following order:
\begin{itemize}
\item[(i)] First, all matched pairs of rooted subtree in which one of the
rooted subtree is in the set $\mathcal{Q}_i[\{u,x\}]$ and the other is in
the set $\mathcal{P}_{i-1}[\{u,x\}]\setminus\mathcal{P}_{i-1}[\{u,v\}]$ are
considered. For every such matched pair, the uncolored rooted subtree is
assigned the same color as its matched colored partner. The number of such
rooted subtrees in the matching $M_{\bar{B}_2}$ is equal to $t_2$.
\item[(ii)] Next, the remaining rooted subtrees from the set
$\mathcal{Q}_i[\{u,x\}]$ are randomly selected one-at-a-time. If the selected
rooted subtree $\vec{R}$ was not matched, and if there is a color that has
already been used previously that can be safely assigned to $\vec{R}$, then
that color is used; otherwise, a new color is used. On the other hand, if the
selected rooted subtree $\vec{R}$ was matched to another rooted subtree
$\vec{S}$, then clearly $\vec{S}$ is also uncolored. In this case both
$\vec{R}$ and $\vec{S}$ are assigned the same color. Again, preference is given
to the colors that are already in use over the use of new colors. According
to Lemma \ref{lem:analysis4}, rooted subtrees in the set
$\mathcal{P}_{i-1}[\{u,w\}]\setminus\left(\mathcal{P}_{i-1}[\{u,v\}]\cup\mathcal{P}_{i-1}[\{u,x\}]\right)$
can never collide with any rooted subtree in the set $\mathcal{Q}_i$.
Therefore, any color used for rooted subtrees in the set
$\mathcal{P}_{i-1}[\{u,w\}]\setminus\left(\mathcal{P}_{i-1}[\{u,v\}]\cup\mathcal{P}_{i-1}[\{u,x\}]\right)$,
that is not used by any other rooted subtree in the set
$\mathcal{P}_{i-1}[\{u,x\}]\setminus\mathcal{P}_{i-1}[\{u,v\}]$, can be
assigned to any of the rooted subtrees in the set $\mathcal{Q}_i$. Let $z_1$ be
the number colors assigned to the rooted subtrees in the set
$\mathcal{P}_{i-1}[\{u,w\}]\setminus\left(\mathcal{P}_{i-1}[\{u,v\}]\cup\mathcal{P}_{i-1}[\{u,x\}]\right)$
that are reused for rooted subtrees in the set $\mathcal{Q}_i[\{u,x\}]$ during
this step of the subroutine. We can bound $z_1$ as
\begin{align}
\label{eq:i}
z_1&\geq\min\Big\{|\mathcal{Q}_i[\{u,x\}]|-t_1-t_2,\\
&|\psi^{\mathrm{GDY}}(\mathcal{P}_{i-1}[\{u,w\}]\setminus\mathcal{P}_{i-1}[\{u,v\}])|\nonumber\\
&{}-|\psi^{\mathrm{GDY}}(\mathcal{P}_{i-1}[\{u,x\}]\setminus\mathcal{P}_{i-1}[\{u,v\}])|\Big\}.\nonumber
\end{align}
The first term in $\min$ is the maximum number of colors required for coloring
all the rooted subtrees in the set $\mathcal{Q}_i[\{u,x\}]$ that remain
uncolored after step (i) described above. The second term is the number of
colors used for assigning colors to the rooted subtrees in the set
$\mathcal{P}_{i-1}[\{u,w\}]\setminus\left(\mathcal{P}_{i-1}[\{u,v\}]\cup\mathcal{P}_{i-1}[\{u,x\}]\right)$
that are not used for any rooted subtree in the set
$\mathcal{P}_{i-1}[\{u,x\}]\setminus\mathcal{P}_{i-1}[\{u,v\}]$.
\item[(iii)] Next, the remaining uncolored rooted subtrees (all the rooted
subtrees in the set $\mathcal{Q}_i\setminus\mathcal{Q}_i[\{u,x\}]$) are
assigned colors one-at-a-time. Again preference is given to the colors that are
already in use over the use of new colors. Since the rooted subtrees in the set
$\mathcal{Q}_i\setminus\mathcal{Q}_i[\{u,x\}]$ can never collide with any
rooted subtree in the set
$\mathcal{P}_{i-1}[\{u,w\}]\setminus\mathcal{P}_{i-1}[\{u,v\}]$, any color used
for rooted subtrees in the set
$\mathcal{P}_{i-1}[\{u,w\}]\setminus\mathcal{P}_{i-1}[\{u,v\}]$ that has not
yet been reused for any rooted subtree in the set $\mathcal{Q}_i[\{u,x\}]$, can
be assigned to any of the rooted subtrees in the set
$\mathcal{Q}_i\setminus\mathcal{Q}_i[\{u,x\}]$. Let $z_2$ be the number of
colors assigned to the rooted subtrees in the set
$\mathcal{P}_{i-1}[\{u,w\}]\setminus\mathcal{P}_{i-1}[\{u,v\}]$ that are reused
for rooted subtrees in the set $\mathcal{Q}_i\setminus\mathcal{Q}_i[\{u,x\}]$
during this step. We can bound $z_2$ as
\begin{align}
\label{eq:ii}
z_2&\geq\min\Big\{|\mathcal{Q}_i\setminus\mathcal{Q}_i[\{u,x\}]|,\\
&|\psi^{\mathrm{GDY}}(\mathcal{P}_{i-1}[\{u,w\}]\setminus\mathcal{P}_{i-1}[\{u,v\}])|-t_2-z_1\Big\}.\nonumber
\end{align}
The first term in $\min$ is the maximum number of colors required for coloring
all the rooted subtrees in the set
$\mathcal{Q}_i\setminus\mathcal{Q}_i[\{u,x\}]$ and the second term is the
number of colors used for coloring the rooted subtrees in the set
$\mathcal{P}_{i-1}[\{u,w\}]\setminus\mathcal{P}_{i-1}[\{u,v\}]$ that have not
yet been reused in the first two steps.
\end{itemize}
Let $z_3$ be the number of colors used for coloring pairs of rooted subtrees in
the set
$\mathcal{P}_{i-1}[\{u,w\}]\setminus\left(\mathcal{P}_{i-1}[\{u,v\}]\cup\mathcal{P}_{i-1}[\{u,x\}]\right)$,
or to pairs of rooted subtrees where one of the rooted subtree belongs to the
set $\mathcal{P}_{i-1}[\{u,x\}]\setminus\mathcal{P}_{i-1}[\{u,v\}]$ and the
other belongs to the set
$\mathcal{P}_{i-1}[\{u,w\}]\setminus\left(\mathcal{P}_{i-1}[\{u,v\}]\cup\mathcal{P}_{i-1}[\{u,x\}]\right)$.
We can determine $z_3$ by subtracting the total number of colors used for
coloring all the rooted subtrees in the set
$\mathcal{P}_{i-1}[\{u,w\}]\setminus\mathcal{P}_{i-1}[\{u,v\}]$ from the sum of
the total number of colors used for coloring all the rooted subtrees in the set
$\mathcal{P}_{i-1}[\{u,x\}]\setminus\mathcal{P}_{i-1}[\{u,v\}]$ and the total
number of rooted subtrees in the set
$\mathcal{P}_{i-1}[\{u,w\}]\setminus\left(\mathcal{P}_{i-1}[\{u,v\}]\cup\mathcal{P}_{i-1}[\{u,x\}]\right)$.
Hence, using equation (\ref{eq:four}),
\begin{align}
\label{eq:iii}
z_3=&|\mathcal{P}_{i-1}[\{u,x\}]\setminus\mathcal{P}_{i-1}[\{u,v\}]|\\
&{}+\big|\mathcal{P}_{i-1}[\{u,w\}]\setminus\left(\mathcal{P}_{i-1}[\{u,v\}]\cup\mathcal{P}_{i-1}[\{u,x\}]\right)\big|\nonumber\\
&{}-t_3-|\psi^{\mathrm{GDY}}(\mathcal{P}_{i-1}[\{u,w\}]\setminus\mathcal{P}_{i-1}[\{u,v\}])|\nonumber\\
=&|\mathcal{Q}_i|-t_3-|\psi^{\mathrm{GDY}}(\mathcal{P}_{i-1}[\{u,w\}]\setminus\mathcal{P}_{i-1}[\{u,v\}])|.\nonumber
\end{align}

We note that the total number of colors required for assigning
colors to all the rooted subtrees in the set
$\mathcal{Q}_i\cup\left(\mathcal{P}_{i-1}[\{u,w\}]\setminus\mathcal{P}_{i-1}[\{u,v\}]\right)$
can be bounded as
\begin{align}
\label{eq:iv}
&\big|\psi^{\mathrm{GDY}}(\mathcal{Q}_i\cup\left(\mathcal{P}_{i-1}[\{u,w\}]\setminus\mathcal{P}_{i-1}[\{u,v\}]\right))\big|\\
=&|\mathcal{Q}_i\cup\left(\mathcal{P}_{i-1}[\{u,w\}]\setminus\mathcal{P}_{i-1}[\{u,v\}]\right)|-|M_{\bar{B}_2}|\nonumber\\
&{}-z_1-z_2-z_3\nonumber\\
\leq&\big|\mathcal{Q}_i\cup\left(\mathcal{P}_{i-1}[\{u,w\}]\setminus\mathcal{P}_{i-1}[\{u,v\}]\right)\big|-|\mathcal{Q}_i|\nonumber\\
&{}+\max\Big\{|\psi^{\mathrm{GDY}}(\mathcal{P}_{i-1}[\{u,w\}]\setminus\mathcal{P}_{i-1}[\{u,v\}])|\nonumber\\
&\phantom{{}+\max\Big\{}-|\mathcal{P}_{i-1}[\{u,w\}]\setminus\mathcal{P}_{i-1}[\{u,v\}]|,\nonumber\\
&\phantom{{}+\max\Big\{}|\psi^{\mathrm{GDY}}(\mathcal{P}_{i-1}[\{u,x\}]\setminus\mathcal{P}_{i-1}[\{u,v\}])|\nonumber\\
&\phantom{{}+\max\Big\{}-|\mathcal{Q}_i\setminus\mathcal{Q}_i[\{u,x\}]|-t_1-t_2,-t_1\Big\}\nonumber\\
\leq&|\mathcal{P}_{i-1}[\{u,w\}]\setminus\mathcal{P}_{i-1}[\{u,v\}]|+\Big[|\mathcal{P}_{i-1}[\{u,x\}]\nonumber\\
&{}\setminus\mathcal{P}_{i-1}[\{u,v\}]|-|\mathcal{Q}_i\setminus\mathcal{Q}_i[\{u,x\}]|-t_1-t_2-t_3\Big]^+\nonumber\\
\leq&|\mathcal{P}_{i-1}[\{u,w\}]\setminus\mathcal{P}_{i-1}[\{u,v\}]|+\left[g-h\right]^+.\nonumber
\end{align}
First inequality is obtained using equations (\ref{eq:four}), (\ref{eq:i}),
(\ref{eq:ii}), (\ref{eq:iii}) and the fact that $|M_{\bar{B}_2}|=t_1+t_2+t_3$.
For getting the second inequality we again use equation (\ref{eq:four}) along
with the facts that the sets $\mathcal{Q}_i$ and
$\mathcal{P}_{i-1}[\{u,w\}]\setminus\mathcal{P}_{i-1}[\{u,v\}]$ are mutually
exclusive, and that the first and the third terms in $\max$ are always less
than or equal to zero and in the second term
$|\psi^{\mathrm{GDY}}(\mathcal{P}_{i-1}[\{u,x\}]\setminus\mathcal{P}_{i-1}[\{u,v\}])|=|\mathcal{P}_{i-1}[\{u,x\}]\setminus\mathcal{P}_{i-1}[\{u,v\}]|-t_3$.
Final inequality uses equations (\ref{eq:four}) and (\ref{eq:three}). Using inequality (\ref{eq:iv})
\begin{align}
\label{eq:v}
&\big|\psi^{\mathrm{GDY}}(\mathcal{P}_{i-1}[\{u,w\}]\cup\mathcal{Q}_i)\big|\\
=&\big|\psi^{\mathrm{GDY}}(\mathcal{P}_{i-1}[\{u,w\}]\cup\mathcal{Q}_i)\big|\nonumber\\
&{}-\big|\psi^{\mathrm{GDY}}(\mathcal{Q}_i\cup\left(\mathcal{P}_{i-1}[\{u,w\}]\setminus\mathcal{P}_{i-1}[\{u,v\}]\right))\big|\nonumber\\
&{}+\big|\psi^{\mathrm{GDY}}(\mathcal{Q}_i\cup\left(\mathcal{P}_{i-1}[\{u,w\}]\setminus\mathcal{P}_{i-1}[\{u,v\}]\right))\big|\nonumber\\
\leq&|\mathcal{P}_{i-1}[\{u,v\}]|+|\mathcal{P}_{i-1}[\{u,w\}]\setminus\mathcal{P}_{i-1}[\{u,v\}]|\nonumber\\
&{}+\left[g-h\right]^+=2l_{\mathcal{R}}+\left[g-h\right]^+.\nonumber
\end{align}
The inequality uses the fact that since the number of colors used for coloring
all the rooted subtrees in the set $\mathcal{P}_{i-1}[\{u,v\}]$ that are
different from the colors used for coloring the rooted subtrees in the set
$\mathcal{Q}_i\cup\left(\mathcal{P}_{i-1}[\{u,w\}]\setminus\mathcal{P}_{i-1}[\{u,v\}]\right)$
is equal to
$|\psi^{\mathrm{GDY}}(\mathcal{P}_{i-1}[\{u,w\}]\cup\mathcal{Q}_i)|-|\psi^{\mathrm{GDY}}(\mathcal{Q}_i\cup\left(\mathcal{P}_{i-1}[\{u,w\}]\setminus\mathcal{P}_{i-1}[\{u,v\}]\right))|$, it is upper bounded by $|\mathcal{P}_{i-1}[\{u,v\}]|$. For the final equality,
we use the fact that the subsets $\mathcal{P}_{i-1}[\{u,v\}]$ and
$\mathcal{P}_{i-1}[\{u,w\}]\setminus\mathcal{P}_{i-1}[\{u,v\}]$ partition the
set $\mathcal{P}_{i-1}[\{u,w\}]=\mathcal{R}[\{u,w\}]$.

Suppose some rooted subtree
$\vec{R}\in\mathcal{P}_{i-1}[\{u,w\}]\setminus\mathcal{P}_{i-1}[\{u,v\}]$
shares its color with another rooted subtree
$\vec{S}\in\mathcal{P}_{i-1}[\{u,v\}]$. In this case, the worst that can happen
is that we may have to add a single new color for coloring all the rooted
subtrees in the set $\mathcal{Q}_i$. On the other hand, rooted subtrees
$\vec{R}\in\mathcal{P}_{i-1}[\{u,w\}]\setminus\mathcal{P}_{i-1}[\{u,v\}],\vec{S}\in\mathcal{P}_{i-1}[\{u,v\}]$
sharing a color means that
$|\psi^{\mathrm{GDY}}(\mathcal{P}_{i-1}[\{u,w\}]\cup\mathcal{Q}_i)|-|\psi^{\mathrm{GDY}}(\mathcal{Q}_i\cup\left(\mathcal{P}_{i-1}[\{u,w\}]\setminus\mathcal{P}_{i-1}[\{u,v\}]\right))|$,
the number of colors used for coloring all the rooted subtrees in the set
$\mathcal{P}_{i-1}[\{u,v\}]$ that are different from the colors used coloring
the rooted subtrees in the set
$\mathcal{Q}_i\cup\left(\mathcal{P}_{i-1}[\{u,w\}]\setminus\mathcal{P}_{i-1}[\{u,v\}]\right)$,
also reduces by $1$. Applying both the observations, we note that the
inequality in (\ref{eq:v}) still holds.
\end{proof}

\subsection{Approximation Ratio}
Using the bounds obtained in Lemmas \ref{lem:analysis2}, \ref{lem:analysis8},
\ref{lem:analysis3} and \ref{lem:analysis5}, we develop the approximation ratio
for GREEDY-COL in the form of a parameterized inequality in Lemma
\ref{lem:analysis6} and then in Lemma \ref{lem:analysis9}, using the valid
ranges of the parameters, we show that the ratio is bounded by $\frac{5}{2}$.

\begin{lemma}
\label{lem:analysis6}
Given a set of rooted subtrees $\mathcal{R}$ on a host tree $T$ of degree at
most $3$, the ratio of the number of colors used by the mapping
$\psi^{\mathrm{GDY}}$ generated by GREEDY-COL and the minimum number of colors
required for coloring all the rooted subtrees in the set $\mathcal{R}$
satisfies
\begin{equation}
\frac{|\psi^{\mathrm{GDY}}(\mathcal{R})|}{\min_{\psi\in\Psi_{\mathcal{R}}}|\psi(\mathcal{R})|}\leq\max_{\alpha,\beta,\gamma,\delta,\epsilon}\frac{2+\min\left\{f_1,\left[f_2-f_3\right]^+\right\}}{2-\min\left\{\beta,\gamma\right\}},\nonumber
\end{equation}
where
\begin{align}
f_1&=\alpha-\left[\beta+\frac{\alpha}{2}-1\right]^+,\nonumber\\
f_2&=\delta+\epsilon-\alpha,\nonumber\\
f_3&=\left[\delta+\frac{\epsilon}{2}+\gamma-2\right]^+,\nonumber
\end{align}
and the maximum is over $\alpha,\beta,\gamma,\delta,\epsilon$
satisfying
\begin{equation}
0\leq\beta,\gamma\leq1,\qquad
0\leq\delta,\epsilon\leq\alpha\leq2,\qquad
\delta+\epsilon\leq2.\nonumber
\end{equation}
\end{lemma}
\begin{proof}
Lemmas \ref{lem:analysis2}, \ref{lem:analysis8}, \ref{lem:analysis3}, and
\ref{lem:analysis5} along with a straightforward induction argument gives that
the number of colors required by GREEDY-COL satisfy
\begin{equation}
\label{eq:eq8}
|\psi^{\mathrm{GDY}}(\mathcal{R})|\leq2l_{\mathcal{R}}+\max_{e_i\in E_T^{\mathrm{4}}}\min\left\{a_i,\left[g_i-h_i\right]^+\right\},
\end{equation}
where $E_T^{\mathrm{4}}\subseteq E_T$ is the set of all the host tree edges of
type $4$ as defined in Lemma \ref{lem:analysis1} encountered in GREEDY-COL and
\begin{align}
a_i&=|\mathcal{Q}_i|-\left(m_{\mathcal{R}[\{u,v\}]}-m_{\mathcal{P}_{i-1}[\{u,v\}]}\right),\nonumber\\
g_i&=|\mathcal{Q}_i[\{u,x\}]|+|\mathcal{P}_{i-1}[\{u,x\}]\setminus\mathcal{P}_{i-1}[\{u,v\}]|-|\mathcal{Q}_i|,\nonumber\\
h_i&=\Bigg[|\mathcal{Q}_i[\{u,x\}]|+\frac{|\mathcal{P}_{i-1}[\{u,x\}]\setminus\mathcal{P}_{i-1}[\{u,v\}]|}{2}\nonumber\\
&\phantom{=\Bigg[}{}+m_{\mathcal{R}[\{u,x\}]}-2l_{\mathcal{R}}\Bigg]^+.\nonumber
\end{align}
Here we follow the naming convention of Lemma \ref{lem:analysis1}, i.e.,
edge $e_i=\{u,v\}$ is the edge being processed in the $i$-th round of coloring
and edges $\{u,w\},\{u,x\}$ have the corresponding meanings as defined in Lemma
\ref{lem:analysis1} whenever $e_i=\{u,v\}$ is of type $4$.

Also, the minimum number of colors required for coloring all the rooted
subtrees in the set $\mathcal{R}$ can be lower bounded as
\begin{align}
\label{eq:eq9}
&\min_{\psi\in\Psi_{\{T,\mathcal{R}\}}}|\psi(\mathcal{R})|\\
\geq&\max_{\{a,b\}\in E_T}\min_{\psi\in\Psi_{\mathcal{R}[\{a,b\}]}}|\psi(\mathcal{R}[\{a,b\}])|\nonumber\\
=&2l_{\mathcal{R}}-\min_{\{a,b\}\in E_T}m_{\mathcal{R}[\{a,b\}]}.\nonumber
\end{align}
The inequality simply says that the number of colors required to color all the
rooted subtrees in $\mathcal{R}$ must be at least as much as the number of
colors required to color the subtrees on every host edge separately. The
equality is due to the fact that $\bar{G}_{\mathcal{R}[\{a,b\}]}$, the
complement of the conflict graph of rooted subtrees on host tree edge
$\{a,b\}$, is bipartite with the size of maximum matching being
$m_{\mathcal{R}[\{a,b\}]}$ and the size of the vertex set being
$|V_{\bar{G}_{\mathcal{R}[\{a,b\}]}}|=|\mathcal{R}[\{a,b\}]|=2l_{\{\vec{T}_H,\mathcal{R}\}}$.
From equations (\ref{eq:eq8}) and (\ref{eq:eq9}) we have
\begin{align}
\label{eq:eq10}
&\frac{|\psi^{\mathrm{GDY}}(\mathcal{R})|}{\min_{\psi\in\Psi_{\mathcal{R}}}|\psi(\mathcal{R})|}\\
\leq&\frac{2l_{\mathcal{R}}+\max_{e_i\in E_T^{\mathrm{4}}}\min\left\{a_i,\left[g_i-h_i\right]^+\right\}}
{2l_{\mathcal{R}}-\min_{\{a,b\}\in E_T}m_{\mathcal{R}[\{a,b\}]}}\nonumber\\
=&\max_{e_i\in E_T^{\mathrm{4}}}\left\{\frac{2l_{\mathcal{R}}+\min\left\{a_i,\left[g_i-h_i\right]^+\right\}}{2l_{\mathcal{R}}-\min_{\{a,b\}\in E_T}m_{\mathcal{R}[\{a,b\}]}}\right\}\nonumber\\
\leq&\max_{e_i\in E_T^{\mathrm{4}}}\left\{\frac{2l_{\mathcal{R}}+\min\left\{a_i,\left[g_i-h_i\right]^+\right\}}{2l_{\mathcal{R}}-\min\left\{m_{\mathcal{R}[\{u,v\}]},m_{\mathcal{R}[\{u,x\}]}\right\}}\right\}.\nonumber
\end{align}

Observe that for any host tree edge $e_i=\{u,v\}$ of type $4$ defined in Lemma
\ref{lem:analysis1}, the following hold.
\begin{itemize}
\item Since $\mathcal{Q}_i\subseteq\mathcal{R}[\{u,v\}]$,
\begin{equation}
|\mathcal{Q}_i|\leq|\mathcal{R}[\{u,v\}]|=2l_{\mathcal{R}}.\nonumber
\end{equation}
Let $|\mathcal{Q}_i|=\alpha_i l_{\mathcal{R}}$, where $\alpha_i$ is a constant
from the set $\left[0,2\right]$.
\item Since $m_{\mathcal{R}[\{u,v\}]}$ is the size of maximum matching in
graph $\bar{G}_{\mathcal{R}[\{u,v\}]}$,
\begin{equation}
m_{\mathcal{R}[\{u,v\}]}\leq\frac{|V_{\bar{G}_{\mathcal{R}[\{u,v\}]}}|}{2}=\frac{|\mathcal{R}[\{u,v\}]|}{2}=l_{\mathcal{R}}.\nonumber
\end{equation}
Let $m_{\mathcal{R}[\{u,v\}]}=\beta_i l_{\mathcal{R}}$, where $\beta_i$
is a constant from the set $\left[0,1\right]$.
\item $\mathcal{R}[\{u,v\}]$, the set of rooted subtrees present on edge
$\{u,v\}$, can be partitioned into $\mathcal{Q}_i$ and
$\mathcal{P}_{i-1}[\{u,v\}]$. Therefore
\begin{equation}
|\mathcal{P}_{i-1}[\{u,v\}]|=|\mathcal{R}[\{u,v\}]|-|\mathcal{Q}_i|=\left(2-\alpha_i\right)l_{\mathcal{R}}.\nonumber
\end{equation}
Since $m_{\mathcal{P}_{i-1}[\{u,v\}]}$ is the size of maximum matching in
graph $\bar{G}_{\mathcal{P}_{i-1}[\{u,v\}]}$, we have
\begin{align}
&m_{\mathcal{P}_{i-1}[\{u,v\}]}\leq\frac{|V_{\bar{G}_{\mathcal{P}_{i-1}[\{u,v\}]}}|}{2}\nonumber\\
=&\frac{|\mathcal{P}_{i-1}[\{u,v\}]|}{2}=\left(1-\frac{\alpha_i}{2}\right)l_{\mathcal{R}}.\nonumber
\end{align}
Also, since $\bar{G}_{\mathcal{P}_{i-1}[\{u,v\}]}$ is a subgraph of
$\bar{G}_{\mathcal{R}[\{u,v\}]}$, we have
\begin{equation}
m_{\mathcal{P}_{i-1}[\{u,v\}]}\leq m_{\mathcal{R}[\{u,v\}]}.\nonumber
\end{equation}
The above two inequalities imply that
\begin{equation}
m_{\mathcal{R}[\{u,v\}]}-m_{\mathcal{P}_{i-1}[\{u,v\}]}\geq\left[\beta_i+\frac{\alpha_i}{2}-1\right]^+l_{\mathcal{R}}.\nonumber
\end{equation}
\item Since $\mathcal{Q}_i[\{u,x\}]\subseteq\mathcal{Q}_i$,
\begin{equation}
|\mathcal{Q}_i[\{u,x\}]|\leq|\mathcal{Q}_i|=\alpha_i l_{\mathcal{R}}.\nonumber
\end{equation}
Let $|\mathcal{Q}_i[\{u,x\}]|=\delta_i l_{\mathcal{R}}$,
where $\delta_i$ is a constant from the set $\left[0,\alpha_i\right]$.
\item $\mathcal{P}_{i-1}[\{u,x\}]\setminus\mathcal{P}_{i-1}[\{u,v\}]$ and
$\mathcal{P}_{i-1}[\{u,v\}]$ are non-overlapping subsets of
$\mathcal{P}_{i-1}[\{u,w\}]=\mathcal{R}[\{u,w\}]$. Also, the set
$\mathcal{R}[\{u,v\}]$ can be partitioned into $\mathcal{Q}_i$ and
$\mathcal{P}_{i-1}[\{u,v\}]$. Therefore,
\begin{align}
&|\mathcal{P}_{i-1}[\{u,x\}]\setminus\mathcal{P}_{i-1}[\{u,v\}]|\nonumber\\
&\leq|\mathcal{R}[\{u,w\}]|-|\mathcal{P}_{i-1}[\{u,v\}]|\nonumber\\
&=|\mathcal{R}[\{u,v\}]|-|\mathcal{P}_{i-1}[\{u,v\}]|\nonumber\\
&=|\mathcal{Q}_i|=\alpha_i l_{\mathcal{R}}.\nonumber
\end{align}
Let
$|\mathcal{P}_{i-1}[\{u,x\}]\setminus\mathcal{P}_{i-1}[\{u,v\}]|=\epsilon_i l_{\mathcal{R}}$,
where $\epsilon_i$ is a constant from the set $\left[0,\alpha_i\right]$.
\item Sets $\mathcal{Q}_i[\{u,x\}]$ and
$\mathcal{P}_{i-1}[\{u,x\}]\setminus\mathcal{P}_{i-1}[\{u,v\}]$ are
non-overlapping subsets of $\mathcal{R}[\{u,x\}]$. Therefore,
\begin{equation}
|\mathcal{Q}_i[\{u,x\}]|+|\mathcal{P}_{i-1}[\{u,x\}]\setminus\mathcal{P}_{i-1}[\{u,v\}]|\leq|\mathcal{R}[\{u,x\}]|.\nonumber
\end{equation}
This implies that $\delta_i+\epsilon_i\leq2$.
\item Since $m_{\mathcal{R}[\{u,x\}]}$ is the size of maximum matching in
graph $\bar{G}_{\mathcal{R}[\{u,x\}]}$,
\begin{equation}
m_{\mathcal{R}[\{u,x\}]}\leq\frac{|V_{\bar{G}_{\mathcal{R}[\{u,x\}]}}|}{2}=\frac{|\mathcal{R}[\{u,x\}]|}{2}=l_{\mathcal{R}}.\nonumber
\end{equation}
Let $m_{\mathcal{R}[\{u,x\}]}=\gamma_i l_{\mathcal{R}}$, where $\gamma_i$
is a constant from the set $\left[0,1\right]$.
\end{itemize}
Combining the above, we get
\begin{align}
\label{eq:eq11}
a_i&\leq\left(\alpha_i-\left[\beta_i+\frac{\alpha_i}{2}-1\right]^+\right)l_{\mathcal{R}},\\
g_i&=\left(\delta_i+\frac{\epsilon_i}{2}-\alpha_i\right)l_{\mathcal{R}},\nonumber\\
h_i&=\left[\delta_i+\frac{\epsilon_i}{2}+\gamma_i-2\right]^+l_{\mathcal{R}},\nonumber
\end{align}
where $\alpha_i,\beta_i,\gamma_i,\delta_i,\epsilon_i$ are known
constants satisfying the following inequalities.
\begin{equation}
\label{eq:eq14} 0\leq\beta_i,\gamma_i\leq1,\qquad
0\leq\delta_i,\epsilon_i\leq\alpha_i\leq2,\qquad
\delta_i+\epsilon_i\leq2
\end{equation}
The lemma follows from equations (\ref{eq:eq10}), (\ref{eq:eq11}) and (\ref{eq:eq14}).
\end{proof}

\begin{lemma}
\label{lem:analysis9}
For any real $\alpha,\beta,\gamma,\delta$ and $\epsilon$ satisfying
\begin{equation}
0\leq\beta,\gamma\leq1,\qquad
0\leq\delta,\epsilon\leq\alpha\leq2,\qquad
\delta+\epsilon\leq2,\nonumber
\end{equation}
and functions $f_1,f_2,f_3$ given by
\begin{align}
f_1&=\alpha-\left[\beta+\frac{\alpha}{2}-1\right]^+,\nonumber\\
f_2&=\delta+\epsilon-\alpha,\nonumber\\
f_3&=\left[\delta+\frac{\epsilon}{2}+\gamma-2\right]^+,\nonumber
\end{align}
the following holds
\begin{equation}
\max_{\alpha,\beta,\gamma,\delta,\epsilon}\frac{2+\min\left\{f_1,\left[f_2-f_3\right]^+\right\}}{2-\min\left\{\beta,\gamma\right\}}\leq\frac{5}{2}.\nonumber
\end{equation}
\end{lemma}
\begin{proof}
Note that for all permissible values of $\alpha,\beta,\gamma,\delta$
and $\epsilon$  we have the following.
\begin{align}
\label{eq:ineq}
&\frac{2+\min\left\{f_1,\left[f_2-f_3\right]^+\right\}}{2-\min\left\{\beta,\gamma\right\}}\\
=&\min\left\{\frac{2+f_1}{2-\min\left\{\beta,\gamma\right\}},\frac{2+[f_2-f_3]^+}{2-\min\left\{\beta,\gamma\right\}}\right\}\nonumber\\
\leq&\min\left\{\frac{2+f_1}{2-\beta},\frac{2+[f_2-f_3]^+}{2-\gamma}\right\}\nonumber
\end{align}
Next observe that for $0\leq\alpha,\beta\leq1$
\begin{align}
\label{eq:ineqa}
&\frac{2+f_1}{2-\beta}
=\frac{2+\alpha-\left[\beta+\frac{1}{2}\alpha-1\right]^+}{2-\beta}\\
=&\frac{2+\alpha-\max\left\{\beta+\frac{1}{2}\alpha-1,0\right\}}{2-\beta}\nonumber\\
=&\frac{\min\left\{3+\frac{1}{2}\alpha-\beta,2+\alpha\right\}}{2-\beta}
\leq\frac{3+\frac{1}{2}\alpha-\beta}{2-\beta}\leq\frac{5}{2}.\nonumber
\end{align}
In case $f_2\leq f_3$, using $0\leq\gamma\leq1$, we have
\begin{equation}
\label{eq:ineqc}
\frac{2+\left[f_2-f_3\right]^+}{2-\gamma}=\frac{2}{2-\gamma}\leq2.
\end{equation}
In case $f_2>f_3$, since $f_3\geq0$ and $f_2=\delta+\epsilon-\alpha>0$,
using $0\leq\gamma\leq1$ and $1\leq\alpha\leq2$, we get
\begin{align}
\label{eq:ineqb}
&\frac{2+\left[f_2-f_3\right]^+}{2-\gamma}=\frac{2+\delta+\epsilon-\alpha-\left[\delta+\frac{1}{2}\epsilon+\gamma-2\right]^+}{2-\gamma}\\
=&\frac{2+\delta+\epsilon-\alpha-\max\left\{\delta+\frac{1}{2}\epsilon+\gamma-2,0\right\}}{2-\gamma}\nonumber\\
=&\frac{\min\left\{2+\frac{1}{2}\epsilon-\alpha+2-\gamma,2+\frac{1}{2}\epsilon-\alpha+\delta+\frac{1}{2}\epsilon\right\}}{2-\gamma}\nonumber\\
=&\frac{2+\frac{1}{2}\epsilon-\alpha}{2-\gamma}+\min\left\{1,\frac{\delta+\frac{1}{2}\epsilon}{2-\gamma}\right\}\nonumber\\
\leq&\frac{2-\frac{1}{2}\alpha}{2-\gamma}+1\leq\frac{5}{2}.\nonumber
\end{align}

From equations (\ref{eq:ineq}), (\ref{eq:ineqa}), (\ref{eq:ineqc}), and
(\ref{eq:ineqb}) we get the required result.
\end{proof}

\begin{theorem}
\label{thm:greedycolor}
GREEDY-COL colors a given set of rooted subtrees on a host tree of degree at most
$3$ using at most $\frac{5}{2}$ times the minimum number of colors.
\end{theorem}
\begin{proof}
The theorem follows from Lemmas \ref{lem:analysis6} and
\ref{lem:analysis9}.
\end{proof}

\section{Discussion and Concluding Remarks}
\label{sec:conclusion}
In this work, motivated by the problem of assigning wavelengths to multicast
traffic requests in all-optical WDM tree networks, we presented Algorithm
\ref{algo:greedycolor} (GREEDY-COL) for coloring a given set of rooted
subtrees of a given host tree with degree at most $3$ with the objective of
minimizing the total number of colors required. We proved that GREEDY-COL is a
$\frac{5}{2}$-approximation algorithm for the problem. Although, we did not
explicitly present the complexity analysis for GREEDY-COL in this paper, it
is straightforward to check that GREEDY-COL runs in polynomial time.

Although the problem is related to the problem of directed path coloring in
trees, the coloring strategy used for that problem is not directly applicable
here. An important difference between the two problems is that if a set of
directed paths collide on some host tree edge, then they must collide on every
host tree edge they share, whereas for rooted subtrees, it is possible for them
to be present on a host tree edge without colliding on that edge but still
collide on some other edge. The implication of this difference is that while in
the case of directed paths, the subproblem of coloring all the paths that share
a host tree vertex is equivalent to edge coloring in a bipartite graph, there
is no such simple equivalence in the case of rooted subtrees. Moreover, the
load of a set of directed paths, which is usually used as the lower bound on
the chromatic number of the corresponding conflict graph, is equal to the
clique number. This is not true in the case of rooted subtrees. In fact, the
lower bound that we employ to determine the approximation ratio for GREEDY-COL,
although better than the load of the set of the rooted subtrees, is still worse
than the clique number of the corresponding conflict graph. One possible approach
to prove a better approximation ratio would be to use the actual clique number of
the conflict graph corresponding to the set of rooted subtrees as the lower bound
for chromatic number.

\bibliography{references}

\end{document}